\documentclass[12pt]{article}%
\usepackage{amsfonts}
\usepackage{amsmath}
\usepackage{amssymb}
\usepackage{graphicx}%
\newtheorem{theorem}{Theorem}

\newtheorem{definition}[theorem]{Definition}

\newtheorem{lemma}[theorem]{Lemma}

\newtheorem{proposition}[theorem]{Proposition}
\newtheorem{remark}[theorem]{Remark}

\newenvironment{proof}[1][Proof]{\noindent\textbf{#1.} }{\ \rule{0.5em}{0.5em}}

\begin{document}

\title{Wigner functions on non-standard symplectic vector spaces}
\author{Nuno Costa Dias\textbf{\thanks{ncdias@meo.pt}}
\and Jo\~{a}o Nuno Prata\textbf{\thanks{joao.prata@mail.telepac.pt
}}} \maketitle

\begin{abstract}
We consider the Weyl quantization on a flat non-standard
symplectic vector space. We focus mainly on the properties of the
Wigner functions defined therein. In particular we show that the
sets of Wigner functions on distinct symplectic spaces are
different but have non-empty intersections. This extends previous
results to arbitrary dimension and arbitrary (constant) symplectic
structure. As a by-product we introduce and prove several concepts
and results on non-standard symplectic spaces which generalize
those on the standard symplectic space, namely the symplectic
spectrum, Williamson's theorem and Narcowich-Wigner spectra. We
also show how Wigner functions on non-standard symplectic spaces
behave under the action of an arbitrary linear coordinate
transformation.
\end{abstract}



\section{Introduction}

In this work we return to the research programme initiated in
\cite{Bastos1,Bastos2}. There we considered the Weyl-Wigner
formulation of quantum mechanics based on a canonical deformation
of the Heisenberg algebra. More specifically, if $z=(x,p) \in
\mathbb{R}^{2n}$ denotes a phase space coordinate of a particle,
then the standard classical Poisson bracket is replaced by the
following deformation
\begin{equation}
\left\{z_{\alpha}, z_{\beta} \right\} = \Omega_{\alpha, \beta}, \hspace{1 cm} 1 \le \alpha, \beta \le 2n,
\label{eqIntroduction1}
\end{equation}
where $(\Omega_{\alpha, \beta})_{1 \le \alpha, \beta \le 2n}$ are the entries of the matrix
\begin{equation}
{\bf \Omega}= \left(
\begin{array}{l l}
{\bf \Theta} & {\bf I}\\
-{\bf I} & {\bf \Upsilon}
\end{array}
\right).
\label{eqNCQM1}
\end{equation}
Here ${\bf I}$ is the $n \times n$ identity matrix and ${\bf
\Theta}=\left(\theta_{ij}\right)_{1\le i,j \le n}, {\bf \Upsilon}=
\left(\eta_{ij}\right)_{1\le i,j \le n}$ are real constant
skew-symmetric $n \times n$ matrices, which measure the extra
noncommutativity in the configuration and momentum spaces,
respectively. One assumes in general that this symplectic form
does not depart appreciably from the standard one \cite{Carroll}:
\begin{equation}
|\theta_{ij} \eta_{kl}| \ll 1,
\label{eqNCQM2}
\end{equation}
for all $i,j,k,l=1, \cdots, n$. Upon quantization, one obtains a deformation of the Heisenberg algebra.

Noncommutative deformations such as this one and others appear in
manifold contexts. They may emerge (i) in the various attempts to
quantize gravity, such as string theory \cite{Seiberg},
noncommutative geometry \cite{Borowiec,Connes1,Madore} and loop
quantum gravity \cite{Rovelli}; (ii) as a means to regularize
quantum field theories \cite{Douglas,Szabo}; (iii) as the result
of quantizing quantum systems with second class constraints
\cite{Henneaux,Nakamura1,Nakamura2}; (iv) in order to stabilize
unstable algebras \cite{Vilela}, (v) as the low-energy physics of
quasicrystals \cite{Monreal}, and quite simply (vi) as a model for
quantum systems under the influence of an external magnetic field
\cite{Delduc}. Regardless of the context and motivation, deforming
the Heisenberg algebra in the realm of non-relativistic quantum
mechanics \cite{Demetrian,Gamboa,Nair} has led to many interesting
and surprising results such as the thermodynamic stability of
black holes \cite{Bastos3}, regularization of black hole
singularities \cite{Bastos4,Bastos5,Nicolini}, modifications of
quantum cosmology scenarios \cite{Obregon,Malekolkalami},
violation of uncertainty principles \cite{Bastos7,Bolonek},
generation of entanglement due to noncommutativity \cite{Bastos6}.
These deformations have also been used in the context of the
Landau problem \cite{Duval,Horvathy} and the quantum Hall effect
\cite{Bellissard}.

We shall follow \cite{Bastos1,Bastos2} and address quantum
mechanics on non-standard symplectic spaces in the framework of
the Weyl-Wigner formulation. As explained in \cite{Bastos2} there
are various aspects of this formulation which make it (in certain
circumstances) more appealing than the ordinary formulation in
terms of operators acting on a Hilbert space. In the Weyl-Wigner
representation one does not have to choose between a position or a
momentum representation of a state, as the Wigner function
provides a joint  distribution of both. It is no surprise that the
uncertainty principle and non-locality take their toll on the
Wigner function in the form of a lack of positivity
\cite{Ellinas,Hudson,Soto}. Nevertheless, Wigner functions have
positive marginals and permit an evaluation of expectation values
of observables with the suggestive formula
\begin{equation}
<\widehat{A}> = \int_{\mathbb{R}^{2n}} a(z) W \rho (z) dz,
\label{eqIntroduction2}
\end{equation}
where $\widehat{A}$ is a self-adjoint operator, $a$ is its Weyl
symbol and $W \rho$ is the Wigner function associated with some
density matrix $\rho$. Moreover, the Weyl-Wigner framework is the
only phase space formulation which is symplectically invariant
\cite{Dias5,Folland,Gosson1,Gosson2,Gosson3,Gosson4,Leray,Wong}.
More precisely, if $W \rho (z)$ is a Wigner function and ${\bf S}$
is some symplectic matrix, then $W \rho({\bf S}z)$ is again a
Wigner function. This makes it a privileged framework to study the
semiclassical limit of quantum mechanics
\cite{Littlejohn,Martinez,Zworski}. The symplectic flavour of this
approach manifests itself equally in the dynamics of a quantum
mechanical system. The time evolution is governed by a quantum
counterpart of the Poisson bracket - called Moyal bracket - which
is a deformation of the Poisson bracket
\cite{Bayen,Fedosov1,Fedosov2,Kontsevich,Moyal,Tosiek,Wilde}. From
this perspective one sometimes calls this phase space formulation
of quantum mechanics {\it deformation quantization}.

The Weyl-Wigner formulation is also particularly useful in the
context of quantum information with continuous variables
\cite{Braunstein,Simon,Werner}.

For the quantization of systems with non-standard symplectic
spaces, the advantage of the Weyl-Wigner formulation becomes more
emphatic. The reason is that the Hilbert space for quantum
mechanics on non-standard symplectic spaces is still $L^2
(\mathbb{R}^n)$. So in the operator formulation, there is no
distinction between states for different symplectic structures. On
the other hand, in deformation quantization we proved in
\cite{Bastos1,Bastos2} that the set of states (Wigner functions)
in noncommutative quantum mechanics differs from the set of
ordinary Wigner functions. We have used this fact as a criterion
for assessing when a transition from noncommutative to ordinary
quantum mechanics has taken place \cite{Dias6}.

In \cite{Bastos2} we considered the set $\mathcal{F}^C$ of
ordinary Wigner functions, the set $\mathcal{F}^{NC}$ of Wigner
functions for the noncommutative symplectic structure
(\ref{eqIntroduction1},\ref{eqNCQM1}), and the set $\mathcal{L}$
of positive (classical) probability densities. We proved that in
{\it dimension $n=2$} any pair of these sets have non-empty
intersections and none of them contains the other. The main
purpose of the present work is to generalize this result to
arbitrary dimension $n$ and to the sets $\mathcal{F}^{\omega_1},
\mathcal{F}^{\omega_2}$ of Wigner functions defined on the
symplectic vector spaces $(\mathbb{R}^{2n}; \omega_1),
(\mathbb{R}^{2n}; \omega_2)$ with symplectic forms $\omega_2 \ne
\pm \omega_1$. Notice that the sympectic forms are completely
arbitrary (not necessarily of the form (\ref{eqNCQM1})), albeit
constant.

As a byproduct we extend a number of definitions and results
concerning ordinary Wigner functions to Wigner functions on
arbitrary, nonstandard symplectic spaces. Moreover, as an
application of our results, we solve the so-called {\it
representability problem} in the case of linear coordinate
transformations. This terminology was coined in
\cite{Cohen,Loughlin} in the context of signal processing for pure
states. Here we use the following definition for an arbitrary
(pure or mixed) state. A given real and measurable phase-space
function $F(z)$ is said to be {\it representable} if there exists
a density matrix $\rho$ such that $F$ is the Wigner function on
some symplectic phase-space associated with $\rho$, i.e. $F=W
\rho$. Here we consider linear coordinate transformations of the
form
\begin{equation}
W \rho (z) \mapsto (\mathcal{U}_{{\bf M}} W \rho) (z) = | \det M| W \rho ({\bf M} z),
\label{eqIntroduction3}
\end{equation}
with ${\bf M} \in Gl(2n; \mathbb{R})$. We then address the problem of the representability of $\mathcal{U}_{{\bf M}} W \rho$.

Our results may also be potentially interesting for quantum
information of continuous variable systems. Indeed, the partial
transpose - a linear coordinate transformation which has as sole
effect the reversal of the momentum of one subsystem - is neither
symplectic nor anti-symplectic \cite{Simon,Werner}. Hence, it maps
a Wigner function on the standard symplectic vector space to a
Wigner function on a non-standard one. This transformation is used
as a criterion to assess whether a state is separable or
entangled.

Here is a summary of the main results of this work.

\vspace{0.3 cm} \noindent {\bf 1)} We define the notion of
Narcowich-Wigner spectrum adapted to an arbitrary symplectic form
(Definition \ref{DefinitionQuantumConditions2}) and prove its main
properties (Theorem \ref{TheoremQuantumConditions5}, Theorem
\ref{TheoremPropertiesNWSpectra}).

\vspace{0.3 cm} \noindent {\bf 2)} We introduce the symplectic
spectrum for an arbitrary symplectic form (Definition
\ref{DefinitionOmegaSymplecticSpectrum1}). We prove a
Williamson-type theorem for non-standard symplectic spaces
(Theorem \ref{TheoremOmegaWilliamsonTheorem1}) and show the
usefulness of symplectic spectra for generalized uncertainty
principles (Theorem \ref{TheoremOmegaRSUP1}). In Theorem
\ref{TheoremIdenticalSpectra1} we show that we can always find a
positive-definite matrix which has distinct symplectic spectra
relative to two different symplectic forms. This result is
fine-tuned for the smallest symplectic eigenvalue in Theorem
\ref{TheoremUnboudedQuotient1}.

\vspace{0.3 cm} \noindent {\bf 3)} We solve the {\it
representability problem} for linear coordinate transformations in
Theorems \ref{TheoremSymplecticCovariance} and
\ref{TheoremRepresentabilityProblem}. In Theorem
\ref{TheoremSymplecticCovariance} we prove that $\mathcal{U}_{{\bf
M}} W \rho$ is representable for any Wigner function $ W \rho$ if
and only if ${\bf M}$ is symplectic or antisymplectic. Another new
result is Theorem \ref{TheoremRepresentabilityProblem}, which
states that if ${\bf M}$ is neither symplectic nor antisymplectic,
then there exists a Wigner function $W \rho$ such that
$\mathcal{U}_{{\bf M}} W \rho$ is again a Wigner function and $W
\rho^{\prime}$ for which $\mathcal{U}_{{\bf M}} W \rho^{\prime}$
is not a Wigner function.

\vspace{0.3 cm} \noindent {\bf 4)} Theorem \ref{TheoremSets1} is
one of the main results of this work and it shows that in the
phase-space formulation different quantizations (i.e. on different
symplectic spaces) yield different sets of states. However, there
are always positive distributions and states that are common to
different quantizations. These results are valid in arbitrary
dimension and for arbitrary flat sympletic vector spaces.

\section*{Notation}

We denote by $\mathbb{P}_{k \times k} (\mathbb{K})$ the convex
cone of real $(\mathbb{K}=\mathbb{R})$ (or complex
$(\mathbb{K}=\mathbb{C})$) symmetric (hermitean) positive
(semi-definite) $k \times k$ matrices, and we write $\mathbb{P}_{k
\times k}^{\times} (\mathbb{K})$ if they are positive-definite.
The set of real skew-symmetric $k \times k$ matrices is $A (k;
\mathbb{R})$. The space of Schwartz test functions is denoted by
$\mathcal{S}(\mathbb{R}^n)$ and its dual
$\mathcal{S}^{\prime}(\mathbb{R}^n)$ are the tempered
distributions. Latin indices $i,j,k, \cdots$, take values in the
set $\left\{1, 2, \cdots, n \right\}$, whereas Greek indices
$\alpha , \beta, \gamma, \cdots $ are phase space indices which
take values in $\left\{1,2, \cdots, 2n \right\}$. The inner
product in $L^2(\mathbb{R}^n)$ is given by
\begin{equation}
<\psi| \phi> = \int_{\mathbb{R}^n} \overline{\psi(x)} \phi (x) dx.
\label{eqNotation1}
\end{equation}
We may at times write $ <\psi| \phi>_{L^2(\mathbb{R}^n)}$ and
$||\psi||_{L^2(\mathbb{R}^n)}$ for the corresponding norm if we
want to emphasize that the functions are defined on the
$n$-dimensional Euclidean space $\mathbb{R}^n$.

The Fourier-Plancherel transform for a function $f \in L^1(\mathbb{R}^n) \cap L^2(\mathbb{R}^n)$ is defined by
\begin{equation}
(\mathcal{F}f) (\omega)= \int_{\mathbb{R}^n} f(x) e^{-i x \cdot \omega} dx
\label{eqNotation2}
\end{equation}
and its inverse is
\begin{equation}
(\mathcal{F}^{-1}f) (x)= \frac{1}{(2 \pi)^n} \int_{\mathbb{R}^n} f(\omega) e^{i x \cdot \omega} d \omega.
\label{eqNotation3}
\end{equation}
In this work, we choose units such that Planck's constant is $h= 2 \pi \hbar= 2 \pi$.

\section{Symplectic vector spaces}

Let $V$ be some real $2n$-dimensional vector space. A symplectic form on $V$ is a skew-symmetric non-degenerate bilinear form, that is a map $\omega:V \times V \to \mathbb{R}$ such that
\begin{equation}
\omega (\alpha x + \beta y,z) =\alpha \omega (x,z) + \beta \omega(y,z), \hspace{1 cm} \omega (x,y) = - \omega(y,x),
\label{eqsymplecticform1}
\end{equation}
for all $x,y,z \in V$ and $\alpha, \beta \in \mathbb{R}$, and
\begin{equation}
\omega (x,y) = 0 , \mbox{ for all } y \in V
\label{eqsymplecticform2}
\end{equation}
if and only if $x=0$.

The symplectic group $Sp (V; \omega)$ is the group of linear automorphisms $\phi: V \to V$ such that
\begin{equation}
\omega \left(\phi (z) , \phi (z^{\prime}) \right) = \omega (z, z^{\prime}),
\label{eqsymplecticform2.1}
\end{equation}
for all $z,z^{\prime} \in V$. An automorphism $\phi$ is said to be anti-symplectic if
\begin{equation}
\omega \left(\phi (z) , \phi (z^{\prime}) \right) = - \omega (z, z^{\prime}),
\label{eqsymplecticform2.2}
\end{equation}
for all $z,z^{\prime} \in V$.
The archetypal symplectic vector space is $V=\mathbb{R}^{2n}$ endowed with the so-called {\it standard symplectic form}. As usual we write $z= (x,p) \in V= \mathbb{R}^n \times (\mathbb{R}^n )^{\ast} \simeq \mathbb{R}^{2n}$, where $x \in \mathbb{R}^n $ is interpreted as the particle's position and $p \in (\mathbb{R}^n)^{\ast} $ as the particle's momentum belonging to the cotangent bundle. The standard symplectic form reads
\begin{equation}
\sigma (z,z^{\prime}) = z \cdot {\bf J^{-1}} z^{\prime} = p \cdot x^{\prime} - x \cdot p^{\prime},
\label{eqsymplecticform3}
\end{equation}
where $z=(x,p)$, $z^{\prime}=(x^{\prime},p^{\prime})$ and ${\bf J}= - {\bf J^{-1}}=-{\bf J^T}$ is the $2n \times 2n$ standard symplectic matrix
\begin{equation}
{\bf J}= \left(
\begin{array}{c c}
{\bf 0} & {\bf I}\\
- {\bf I} & {\bf 0}
\end{array}
\right).
\label{eqsymplecticform4}
\end{equation}
The symplectic group $Sp(\mathbb{R}^{2n}; \sigma) \equiv Sp(n; \sigma)$ is then the set of matrices ${\bf P}$ such that
\begin{equation}
{\bf P} {\bf J} {\bf P^T} = {\bf J}.
\label{eqsymplecticform4.1}
\end{equation}
We remark that if ${\bf P} \in Sp (n; \sigma)$, then also ${\bf P^{-1}}, {\bf P^T} \in Sp (n; \sigma)$.

A matrix ${\bf A}$ is $\sigma$-anti-symplectic \cite{Dias5} if
\begin{equation}
{\bf A} {\bf J} {\bf A^T} = - {\bf J}.
\label{eqsymplecticform4.2}
\end{equation}
All $\sigma$-anti-symplectic matrices ${\bf A}$ can be written as
\begin{equation}
{\bf A} = {\bf R}{\bf P},
\label{eqsymplecticform4.3}
\end{equation}
where ${\bf P} \in Sp (n; \sigma)$ and
\begin{equation}
{\bf R} =\left(
\begin{array}{c c}
{\bf I} & {\bf 0}\\
{\bf 0} & - {\bf I}
\end{array}
\right).
\label{eqsymplecticform4.4}
\end{equation}
Thus a $\sigma$-anti-symplectic transformation is a $\sigma$-symplectic transformation followed by a reflection of the particle's momentum.

More generally, for any other symplectic form on $\mathbb{R}^{2n}$, there exists a real skew-symmetric matrix ${\bf \Omega} \in Gl(2n)$ such that
\begin{equation}
\omega (z,z^{\prime}) = z \cdot {\bf \Omega^{-1}} z^{\prime}, \hspace{1 cm} z,z^{\prime} \in \mathbb{R}^{2n}.
\label{eqsymplecticform5}
\end{equation}
A well known theorem in symplectic geometry \cite{Cannas,Gosson1}
states that all symplectic vector spaces of equal dimension are
symplectomorphic. In other words, if $(V, \omega)$ and
$(V^{\prime}, \omega^{\prime})$ have the same dimension, then
there exists a linear isomorphism $\phi: V \to V^{\prime}$ such
that
\begin{equation}
\phi^{\ast} \omega^{\prime} = \omega.
\label{eqsymplecticform6}
\end{equation}
Here $\phi^{\ast}$ denotes the pull-back of $\phi$. In particular, $(\mathbb{R}^{2n}, \sigma)$ is symplectomorphic with $(\mathbb{R}^{2n}, \omega)$ for any symplectic form (\ref{eqsymplecticform5}). So if $\phi:(\mathbb{R}^{2n}, \sigma) \to (\mathbb{R}^{2n}, \omega)$ is given by
\begin{equation}
\phi(z) = {\bf S} z,
\label{eqsymplecticform7}
\end{equation}
for some ${\bf S} \in Gl(2n)$, then we have from (\ref{eqsymplecticform6}) that
\begin{equation}
\phi^{\ast} \omega = \sigma \Leftrightarrow \omega({\bf S} z, {\bf S} z^{\prime}) = \sigma (z,z^{\prime}),
\label{eqsymplecticform8}
\end{equation}
for all $z, z^{\prime} \in \mathbb{R}^{2n}$. And thus
\begin{equation}
{\bf \Omega} = {\bf S} {\bf J} {\bf S^T}.
\label{eqsymplecticform9}
\end{equation}

The matrix ${\bf S}$ in (\ref{eqsymplecticform9}) is not unique. Indeed, we have
\begin{lemma}\label{LemmaDarbouxmap1}
Suppose that ${\bf S}$ is a solution of (\ref{eqsymplecticform9}). Then ${\bf S^{\prime}}$ is also a solution if and only if $ {\bf S^{\prime -1} S}, {\bf S^{-1} S^{\prime}} \in Sp(n; \sigma)$.
\end{lemma}

\begin{proof}
We have ${\bf \Omega} = {\bf S} {\bf J} {\bf S^T}={\bf S^{\prime}} {\bf J} {\bf S^{\prime T}}$ if and only if ${\bf S^{\prime -1} S}{\bf J} ({\bf S^{\prime -1} S})^T = {\bf S^{-1} S^{\prime}} {\bf J} ({\bf S^{-1} S^{\prime}})^T = {\bf J}$.
\end{proof}

We denote by $\mathcal{D} (n;\omega)$ the set of all real $2n \times 2n$ matrices satisfying (\ref{eqsymplecticform9}) and call them {\it Darboux matrices}. The corresponding symplectic automorphism (\ref{eqsymplecticform7}) is called a {\it Darboux map}.

A matrix ${\bf M}$ is $\omega$-symplectic if
\begin{equation}
{\bf M} {\bf \Omega} {\bf M^T}= {\bf \Omega},
\label{eqomegasymplectic1}
\end{equation}
and it is $\omega$-anti-symplectic if
\begin{equation}
{\bf M} {\bf \Omega} {\bf M^T}= - {\bf \Omega}.
\label{eqomegasymplectic2}
\end{equation}
Any $\omega$-symplectic (resp. $\omega$-anti-symplectic) matrix ${\bf M}$ is of the form
\begin{equation}
{\bf M}={\bf S}{\bf P} {\bf S^{-1}},
\label{eqomegasymplectic3}
\end{equation}
where ${\bf S} \in \mathcal{D}(n; \omega)$ and ${\bf P}$ is a $\sigma$-symplectic (resp. $\sigma$-anti-symplectic) matrix.

For future reference we consider the following Proposition.

\begin{proposition}\label{PropositionSymplecticInequality}
Let $\eta $ be an arbitrary symplectic form on $\mathbb{R}^{2n}$ such that
\begin{equation}
|\eta (z,z^{\prime})| \le |\sigma (z,z^{\prime})|,
\label{eqSymplecticInequality1}
\end{equation}
for all $z,z^{\prime} \in \mathbb{R}^{2n}$. Then, there exists a constant $0 < |a| \le 1$ such that
\begin{equation}
\eta = a \sigma.
\label{eqSymplecticInequality2}
\end{equation}
\end{proposition}

\begin{proof}
Let $\eta (z,z^{\prime}) = z^T {\bf \Sigma^{-1}} z^{\prime}$ for $z,z^{\prime} \in \mathbb{R}^{2n}$ and ${\bf \Sigma} \in Gl(2n; \mathbb{R})$. Let $\left\{\delta_{\alpha} \right\}_{1 \le \alpha \le 2n}$ denote the canonical basis of $\mathbb{R}^{2n}$. If we set $z=\delta_{\alpha}$ and $z^{\prime} = \delta_{\beta}$ in (\ref{eqSymplecticInequality1}), we obtain:
\begin{equation}
|\Sigma_{\alpha \beta}^{-1}| \le |J_{\alpha \beta}|,
\label{eqSymplecticInequality3}
\end{equation}
for all $1 \le \alpha, \beta \le 2n$. So $\Sigma_{\alpha \beta}^{-1}=0$, whenever $J_{\alpha \beta}=0$. And so ${\bf \Sigma^{-1}}$ must be of the form
\begin{equation}
{\bf \Sigma^{-1}} = \left(
\begin{array}{c c}
{\bf 0} & {\bf D}\\
- {\bf D} & {\bf 0}
\end{array}
\right),
\label{eqSymplecticInequality4}
\end{equation}
where ${\bf D} = diag (\mu_1, \cdots, \mu_n)$. Next set $z=\delta_i -\delta_j$ and $z^{\prime} = \delta_{n+i}+ \delta_{n+j}$, for $1 \le i < j \le n$. A simple calculation shows that
\begin{equation}
\sigma (z,z^{\prime})=0, \hspace{1 cm} \eta (z,z^{\prime})= \mu_i - \mu_j.
\label{eqSymplecticInequality5}
\end{equation}
But from (\ref{eqSymplecticInequality1}) it follows that $\mu_i = \mu_j$ for all $1 \le i < j \le n$. And so there exists $a \ne 0$ such that (\ref{eqSymplecticInequality2}) holds. But again from (\ref{eqSymplecticInequality1}), we must have $0 < |a| \le 1$.
\end{proof}

\section{Weyl Quantization}

\subsection{Weyl quantization on $(\mathbb{R}^{2n}; \sigma)$}

One of the basic ingredients for the quantization of a classical system defined on the standard symplectic vector space $(\mathbb{R}^{2n}; \sigma)$ is the Heisenberg-Weyl algebra:
\begin{equation}
\left[\widehat{Z}_{\alpha}, \widehat{Z}_{\beta} \right] = i J_{\alpha \beta} \widehat{I}, \hspace{0.5 cm} \left[\widehat{Z}_{\alpha}, \widehat{I} \right] = 0, \hspace{1 cm} \alpha, \beta =1, \cdots, 2n,
\label{eqHeisenbergWeylAlgebra1}
\end{equation}
where $\widehat{Z} = (\widehat{q}, \widehat{p})$ are the quantum-mechanical counterparts of the classical phase-space variables $z=(x,p)$ and $\widehat{I}$ is the identity operator. Their action on functions in their (densely defined) domains in $L^2 (\mathbb{R}^n)$ is given by
\begin{equation}
(\widehat{I} \psi) (x) =  \psi (x), \hspace{1 cm}(\widehat{q}_i \psi) (x) = x_i \psi (x), \hspace{1 cm} (\widehat{p}_j \psi) (x) = -i \partial_j \psi (x),
\label{eqHeisenbergWeylAlgebra2}
\end{equation}
for $i,j=1, \cdots, n$. These operators are not bounded, which prevents the construction of a $C^*$-algebra of observables. A familiar way to circumvent this is to consider the alternative algebra of Heisenberg-Weyl displacement operators:
\begin{equation}
\widehat{D^{\sigma}} (\xi) = e^{ i \sigma (\xi,\widehat{Z})}.
\label{eqHeisenbergWeylOperator1}
\end{equation}
The action of $\widehat{D^{\sigma}} (\xi)$ on $\psi \in \mathcal{S} (\mathbb{R}^n)$ is given by the explicit formula
\begin{equation}
\widehat{D^{\sigma}} (\xi) \psi (x) = e^{ip_0 \cdot x - \frac{i}{2} p_0 \cdot x_0} \psi (x-x_0),
\label{eqHeisenbergWeylOperator2}
\end{equation}
for $\xi = (x_0,p_0)$. This extends to a unitary operator from $L^2 (\mathbb{R}^n)$ to $L^2 (\mathbb{R}^n)$ and to a continuous operator from $\mathcal{S}^{\prime} (\mathbb{R}^n)$ to $\mathcal{S}^{\prime} (\mathbb{R}^n)$.

Their commutation relations are easily established from (\ref{eqHeisenbergWeylOperator2}) or, heuristically, from (\ref{eqHeisenbergWeylAlgebra1},\ref{eqHeisenbergWeylOperator1}) using the Baker-Campbell-Hausdorff formula:
\begin{equation}
\widehat{D^{\sigma}} (\xi)\widehat{D^{\sigma}} (\zeta) = e^{\frac{i}{2} \sigma (\xi, \zeta)} \widehat{D^{\sigma}} (\xi + \zeta) = e^{i \sigma (\xi, \zeta)} \widehat{D^{\sigma}} (\zeta) \widehat{D^{\sigma}} (\xi),
\label{eqCommutationRelationsWeylOps1}
\end{equation}
for $\xi, \zeta \in \mathbb{R}^{2n}$. Operators of the form
\begin{equation}
\widehat{U}^{\sigma} (\xi, \tau) =e^{i \tau} \widehat{D^{\sigma}} (\xi)
\label{eqCommutationRelationsWeylOps2}
\end{equation}
constitute an irreducible unitary representation of the Heisenberg group $\mathbb{H} (n)$ of elements $(\xi, \tau) \in \mathbb{R}^{2n} \times \mathbb{R}$ with group multiplication
\begin{equation}
(\xi, \tau) \cdot ( \xi^{\prime}, \tau^{\prime}) = \left(\xi+ \xi^{\prime},\tau + \tau^{\prime} + \frac{1}{2} \sigma (\xi, \xi^{\prime}) \right).
\label{eqHeisenberggroup1}
\end{equation}
This is called the Schr\"odinger representation and, according to the Stone-von Neumann theorem \cite{Reed}, it is in fact the only unitary irreducible representation of $\mathbb{H} (n)$ on $L^2 (\mathbb{R}^n)$ up to rescalings of Planck's constant.

Given a tempered distribution $a^{\sigma} \in \mathcal{S}^{\prime}
(\mathbb{R}^{2n})$ \cite{Grubb,Hormander}, the Weyl operator with
symbol $a^{\sigma}$ is the Bochner integral
\cite{Dubin,Folland,Gosson1,Pool,Wong}
\begin{equation}
\widehat{A} := \left(\frac{1}{2 \pi} \right)^n \int_{\mathbb{R}^{2n}} \mathcal{F}_{\sigma} a^{\sigma} (z) \widehat{D^{\sigma}} (z) d z.
\label{eqWeylOperator1}
\end{equation}
Here $\mathcal{F}_{\sigma} a$ denotes the symplectic Fourier transform
\begin{equation}
\mathcal{F}_{\sigma} a (z ) := \left(\frac{1}{2 \pi} \right)^n \int_{\mathbb{R}^{2n}} e^{-i \sigma (z,z^{\prime})} a(z^{\prime}) dz^{\prime}.
\label{eqSymplecticFourierTransform1}
\end{equation}
We note that $\mathcal{F}_{\sigma}$ in (\ref{eqSymplecticFourierTransform1}) extends into an involutive automorphism $\mathcal{S}^{\prime} (\mathbb{R}^n)\rightarrow \mathcal{S}^{\prime} (\mathbb{R}^n)$.

The Weyl correspondence, written
$\widehat{A}\overset{\mathrm{Weyl}}{\longleftrightarrow}a^{\sigma}$
or $a^{\sigma} \overset{\mathrm{Weyl}}{\longleftrightarrow}
\widehat{A}$ between a symbol $a^{\sigma} \in \mathcal{S}^{\prime}
(\mathbb{R}^{2n})$ and the associated Weyl operator $\widehat{A}$
is a one-to-one map from $\mathcal{S}^{\prime} (\mathbb{R}^{2n})$
onto the space $\mathcal{L} \left(\mathcal{S}(\mathbb{R}^n),
\mathcal{S}^{\prime} (\mathbb{R}^n) \right)$ of bounded linear
operators $\mathcal{S} (\mathbb{R}^n) \to \mathcal{S}^{\prime}
(\mathbb{R}^n)$. This can be proven using Schwartz's kernel
Theorem and the fact that the Weyl symbol $a^{\sigma}$ of the
operator $\widehat{A}$ is related with the distributional kernel
$K_a$ of that operator by a partial Fourier transform
\cite{Bracken,Gosson1,Wong}:
\begin{equation}
a^{\sigma}(x,p)= \int_{\mathbb{R}^n} e^{-i y \cdot p} K_a \left(x + \frac{y}{2},x - \frac{y}{2} \right) dy,
\label{eqWeylOperator2}
\end{equation}
where $K_a \in  \mathcal{S}^{\prime} (\mathbb{R}^{2n})$ and the Fourier transform is defined in the usual distributional sense.

Conversely, by the Fourier inversion theorem, the kernel $K_a$ can be expressed in terms of the symbol $a^{\sigma}$:
\begin{equation}
K_a (x,y) =  \left(\frac{1}{2 \pi}\right)^n \int_{\mathbb{R}^n} e^{i p \cdot (x-y)} a^{\sigma} \left( \frac{x+y}{2} , p \right) dp.
\label{eqWeylOperator3}
\end{equation}

A Weyl operator is formally self-adjoint if and only if its symbol $a^{\sigma}$ is real.

Now, let $\psi \in L^2 (\mathbb{R}^n)$ and define the rank-one operator $\widehat{\rho_{\psi}}$
\begin{equation}
\widehat{\rho_{\psi}} \phi = <\psi| \phi> \psi,
\label{eqWeylOperator4}
\end{equation}
for all $\phi \in L^2 (\mathbb{R}^n)$. This is a Hilbert-Schmidt operator with kernel:
\begin{equation}
K_{\rho_{\psi}} (x,y) = \psi (x) \overline{\psi (y)} .
\label{eqWeylOperator4}
\end{equation}
The corresponding Weyl symbol is (cf.(\ref{eqWeylOperator2})):
\begin{equation}
\rho^{\sigma}_{\psi} (x,p) = \int_{\mathbb{R}^n} e^{-i y \cdot p} \psi \left(x+ \frac{y}{2} \right) \overline{\psi \left(x- \frac{y}{2} \right)} dy
\label{eqWeylOperator5}
\end{equation}
which is proportional to the celebrated Wigner function \cite{Wigner}:
\begin{equation}
W^{\sigma} \psi (x,p) := \left(\frac{1}{2 \pi } \right)^n \rho_{\psi}^{\sigma} (x,p).
\label{eqWeylOperator6}
\end{equation}
The multiplicative constant $(2 \pi)^{-n}$ is included to ensure the correct normalization
\begin{equation}
\int_{\mathbb{R}^{2n}} W^{\sigma} \psi (z) dz =1,
\label{eqWeylOperator7}
\end{equation}
whenever $||\psi||_{L^2}=1$.

The operator $\widehat{\rho_{\psi}}$ is the density matrix of the pure state $\psi$. Density matrices  provide a unified formulation of pure and mixed states, playing a key role in many different contexts, notably semi-classical limit, decoherence, {\it etc} \cite{Giulini}. Let then $\left\{p_{\alpha} \right\}_{\alpha}$ be some probability distribution
\begin{equation}
0 \le p_{\alpha} \le 1, \hspace{1 cm} \sum_{\alpha} p_{\alpha} =1,
\label{eqWeylOperator8}
\end{equation}
and $\left\{\widehat{\rho}_{\alpha} \right\}_{\alpha}$ a collection of pure state density matrices associated with normalized states $\psi_{\alpha} \in L^2 (\mathbb{R}^n)$. The statistical mixture
\begin{equation}
\widehat{\rho} = \sum_{\alpha} p_{\alpha} \widehat{\rho}_{\alpha}
\label{eqWeylOperator9}
\end{equation}
is called the density matrix of a mixed state. One needs some caution in interpreting eq.(\ref{eqWeylOperator9}). The usual statement that the system is in state $\widehat{\rho}_{\alpha}$ with probability $ p_{\alpha}$ is not very rigorous, because there are (in general infinitely) many different decompositions of a density matrix of the form (\ref{eqWeylOperator9}).

An useful way to tell pure states from mixed states is to compute the so-called {\it purity} of the state. Moyal's identity \cite{Moyal} states that if $\psi , \phi \in L^2 (\mathbb{R}^n)$, then $W^{\sigma} \psi , W^{\sigma} \phi \in  L^2 (\mathbb{R}^{2n})$ and that:
\begin{equation}
<W^{\sigma} \psi|W^{\sigma} \phi>_{L^2 (\mathbb{R}^{2n})} =\frac{1}{(2 \pi)^n} \left|<\psi| \phi> _{L^2 (\mathbb{R}^n)} \right|^2.
\label{eqMoyalidentity1}
\end{equation}
In particular:
\begin{equation}
(2 \pi)^n||W^{\sigma} \psi||_{L^2 (\mathbb{R}^{2n})}^2 = ||\psi||_{L^2 (\mathbb{R}^n)}^4 =1.
\label{eqMoyalidentity2}
\end{equation}
On the other hand, if we have a mixed state $W^{\sigma} \rho $, then:
\begin{equation}
(2 \pi)^n||W^{\sigma} \rho||_{L^2 (\mathbb{R}^{2n})}^2 <1.
\label{eqMoyalidentity3}
\end{equation}
The {\it purity} of a state $W^{\sigma} \rho$ is defined by
\begin{equation}
0 < \mathcal{P} \left[W^{\sigma} \rho \right] = (2 \pi)^n || W^{\sigma} \rho||_{L^2 (\mathbb{R}^{2n})}^2  \le 1,
\label{eqMoyalidentity3}
\end{equation}
and we have
\begin{equation}
\begin{array}{l l}
\mathcal{P} \left[W^{\sigma} \rho \right] =1 & \mbox{if } W^{\sigma} \rho \mbox{ is a pure state}\\
& \\
\mathcal{P} \left[W^{\sigma} \rho \right] <1 & \mbox{if } W^{\sigma} \rho \mbox{ is a mixed state}
\end{array}
\label{eqMoyalidentity4}
\end{equation}

It is a well known fact from the theory of compact operators \cite{Reed} that an operator $\widehat{\rho}$ can be written as in (\ref{eqWeylOperator8},\ref{eqWeylOperator9}) if and only if it is a positive and trace-class operator with unit trace. If $\rho(x,y)$ is the corresponding kernel then the associated Wigner function is
\begin{equation}
\begin{array}{c}
W^{\sigma} \rho (x,p) = \left(\frac{1}{2 \pi}\right)^n \int_{\mathbb{R}^n}  e^{-i y \cdot p} \rho \left(x+ \frac{y}{2} ,x- \frac{y}{2} \right) dy =\\
\\
= \sum_{\alpha} p_{\alpha} W^{\sigma} \psi_{\alpha} (x,p),
\end{array}
\label{eqWeylOperator10}
\end{equation}
with uniform convergence. If $\widehat{A}\overset{\mathrm{Weyl}}{\longleftrightarrow}a^{\sigma}$ is some self-adjoint operator such that $\widehat{A} \widehat{\rho}$ is trace-class then the expectation value of $\widehat{A}$ in the state $\widehat{\rho}$ can be evaluated according to the following remarkable formula.
\begin{equation}
<\widehat{A}>_{\widehat{\rho}} = Tr (\widehat{A}\widehat{\rho}) = \int_{\mathbb{R}^{2n}} a^{\sigma} (z) W^{\sigma} \rho (z) dz.
\label{eqWeylOperator11}
\end{equation}
Because of this formula and of the fact that Wigner functions have positive marginal distributions, one is tempted to interpret them as joint position and momentum probability densities. However, this is precluded by their lack of positivity as stated in Hudson's Theorem \cite{Ellinas,Hudson,Soto}.

\subsection{Weyl quantization on $(\mathbb{R}^{2n}; \omega)$}

Alternatively, one may choose to quantize the system on a non-standard symplectic vector space $(\mathbb{R}^{2n}; \omega)$ with symplectic form $\omega$ given by (\ref{eqsymplecticform5}) \cite{Bastos1,Bastos2,Chowdhuri1,Chowdhuri2,Dias1,Dias2}.

Upon quantization, the Heisenberg-Weyl algebra (\ref{eqHeisenbergWeylAlgebra1}) is replaced by the algebra
\begin{equation}
\left[\widehat{\Xi}_{\alpha},\widehat{\Xi}_{\beta} \right] = i \Omega_{\alpha \beta} \widehat{I}, \hspace{0.5 cm} \left[\widehat{\Xi}_{\alpha},\widehat{I}\right] = 0, \hspace{1 cm} \alpha, \beta =1, \cdots, 2n.
\label{eqNonstandardWeylalgebra1}
\end{equation}
In the sequel, we shall assume that
\begin{equation}
\det {\bf \Omega} =1 , \hspace{1 cm} \det {\bf S}= \pm 1,
\label{eqNonstandardWeylalgebra1.1}
\end{equation}
for all ${\bf S} \in \mathcal{D}(n; \omega)$. This is because different values of the determinant of ${\Omega}$ simply amount to a rescaling of the variables $\left\{\widehat{\Xi}_{\alpha} \right\}_{1 \le \alpha \le 2n}$ (or, equivalently of a rescaling of Planck's constant).

Upon exponentiation of (\ref{eqNonstandardWeylalgebra1}), we obtain the non-standard Heisenberg-Weyl displacement operators:
\begin{equation}
\widehat{D}^{\omega} (\xi) = e^{i \omega (\xi, \widehat{\Xi})},
\label{eqNonstandardWeylalgebra2}
\end{equation}
satisfying the commutation relations
\begin{equation}
\widehat{D}^{\omega} (\xi ) \widehat{D}^{\omega} (\zeta) = e^{\frac{i}{2} \omega (\xi, \zeta)} \widehat{D}^{\omega} (\xi +  \zeta) = e^{i \omega (\xi, \zeta)}  \widehat{D}^{\omega} (\zeta) \widehat{D}^{\omega} (\xi) .
\label{eqNonstandardWeylalgebra3}
\end{equation}
As before, the elements of the form
\begin{equation}
\widehat{U}^{\omega} (\xi, \tau) = e^{i \tau} \widehat{D}^{\omega} (\xi ) ,
\label{eqNonstandardWeylalgebra4}
\end{equation}
constitute a unitary irreducible representation of the modified Heisenberg group $\mathbb{H}^{\omega} (n)$ with multiplication
\begin{equation}
(\xi, \tau) \cdot (\xi^{\prime}, \tau^{\prime}) = \left(\xi + \xi^{\prime}, \tau + \tau^{\prime} + \frac{1}{2} \omega (\xi, \xi^{\prime}) \right).
\label{eqNonstandardWeylalgebra5}
\end{equation}
Notice that if ${\bf S} \in \mathcal{D}(n; \omega)$ and $\widehat{Z}$ obey the Heisenberg-Weyl algebra (\ref{eqHeisenbergWeylAlgebra1}), then
\begin{equation}
\widehat{\Xi} = {\bf S}\widehat{Z}
\label{eqNonstandardWeylalgebra6}
\end{equation}
satisfy (\ref{eqNonstandardWeylalgebra1}). Consequently, we may write
\begin{equation}
 \widehat{D}^{\omega} (\xi )= e^{i \omega(\xi, \widehat{\Xi})} = e^{i \omega(\xi, {\bf S} \widehat{Z})}=e^{i \sigma({\bf S^{-1}} \xi, \widehat{Z})} = \widehat{D}^{\sigma} ({\bf S^{-1}} \xi ).
\label{eqNonstandardWeylalgebra6}
\end{equation}
From this relation and eq.(\ref{eqCommutationRelationsWeylOps1}), one easily proves the commutation relations (\ref{eqNonstandardWeylalgebra3}).

Also, if we perform the transformation $\xi \mapsto z= {\bf S} \xi$ with ${\bf S} \in \mathcal{D}(n; \omega)$ in (\ref{eqWeylOperator1}) and substitute (\ref{eqNonstandardWeylalgebra6}), we obtain:
\begin{equation}
\widehat{A} = \left(\frac{1}{2 \pi} \right)^n  \int_{\mathbb{R}^{2n}} \mathcal{F}_{\sigma} a^{\sigma} ({\bf S^{-1}} \xi)  \widehat{D}^{\omega} (\xi) d \xi
\label{eqNonstandardWeylalgebra7}
\end{equation}
Next notice that
\begin{equation}
\begin{array}{c}
\mathcal{F}_{\sigma} a ({\bf S^{-1}} \xi) = \left(\frac{1}{2 \pi} \right)^n \int_{\mathbb{R}^{2n}}e^{-i \sigma ({\bf S^{-1}} \xi, z)} a (z) d z = \\
\\
=\left(\frac{1}{2 \pi} \right)^n \int_{\mathbb{R}^{2n}}e^{-i \omega ( \xi, {\bf S}z)} a (z) d z=\\
\\
= \left(\frac{1}{2 \pi} \right)^n  \int_{\mathbb{R}^{2n}} e^{-i \omega ( \xi, \xi^{\prime})} a ({\bf S^{-1}} \xi^{\prime}) d \xi^{\prime} .
\end{array}
\label{eqNonstandardWeylalgebra8}
\end{equation}
This then suggests the following definition of the $\omega$-symplectic Fourier transform \cite{Dias2}:
\begin{equation}
\mathcal{F}_{\omega} a(\xi) := \frac{1}{(2 \pi)^n } \int_{\mathbb{R}^{2n}} e^{-i \omega ( \xi, \xi^{\prime})} a(\xi^{\prime}) d \xi^{\prime},
\label{eqNonstandardWeylalgebra9}
\end{equation}
for any $a \in \mathcal{S}^{\prime} (\mathbb{R}^{2n})$. With this choice of normalization $\mathcal{F}_{\omega}$ is involutive $\mathcal{F}_{\omega}\mathcal{F}_{\omega} a= a$.

Also we define the $\omega$-Weyl symbol of the operator $\widehat{A}$:
\begin{equation}
a^{\omega} (\xi) := a^{\sigma} ({\bf S^{-1}} \xi).
\label{eqNonstandardWeylalgebra10}
\end{equation}
Altogether, from (\ref{eqNonstandardWeylalgebra7}-\ref{eqNonstandardWeylalgebra10}), we obtain:
\begin{equation}
\widehat{A} = \frac{1}{(2 \pi)^n  } \int_{\mathbb{R}^{2n}} \mathcal{F}_{\omega} a^{\omega} (\xi) \widehat{D}^{\omega} (\xi) d \xi,
\label{eqNonstandardWeylalgebra11}
\end{equation}
which defines the correspondence principle between an operator $\widehat{A}$ and its $\omega$-Weyl symbol $a^{\omega}$. From eq.(\ref{eqWeylOperator6},\ref{eqWeylOperator11}) it follows that
\begin{equation}
\begin{array}{c}
<\widehat{A} >_{\widehat{\rho} } = \frac{1}{(2 \pi)^n} \int_{\mathbb{R}^{2n}} a^{\sigma} (z) \rho^{\sigma} (z) dz=\\
\\
= \frac{1}{(2 \pi)^n  } \int_{\mathbb{R}^{2n}} a^{\sigma} ({\bf S^{-1}} \xi) \rho^{\sigma} ({\bf S^{-1}} \xi ) d \xi=\\
\\
= \frac{1}{(2 \pi)^n  } \int_{\mathbb{R}^{2n}}  a^{\omega} (\xi) \rho^{\omega} ( \xi) d \xi =\\
\\
=  \int_{\mathbb{R}^{2n}}  a^{\omega} (\xi ) W^{\omega} \rho (\xi ) d \xi,
\end{array}
\label{eqNonstandardWeylalgebra12}
\end{equation}
where
\begin{equation}
 W^{\omega} \rho (\xi) =  W^{\sigma} \rho ({\bf S^{-1}} \xi),
\label{eqNonstandardWeylalgebra13}
\end{equation}
is the $\omega$-Wigner function associated with the density matrix $\widehat{\rho}$ \cite{Bastos1,Bastos2,Chowdhuri1,Chowdhuri2,Jiang}.

Regarding the purity of the states on the non-standard symplectic space $(\mathbb{R}^{2n}; \omega)$ we have:
\begin{equation}
\begin{array}{c}
\mathcal{P} \left[ W^{\omega} \rho \right] = (2 \pi)^n \int_{\mathbb{R}^{2n}} | W^{\omega} \rho(\xi)|^2 d \xi =(2 \pi)^n \int_{\mathbb{R}^{2n}} | W^{\sigma} \rho({\bf S^{-1}} \xi)|^2 d \xi = \\
\\
= (2 \pi)^n \int_{\mathbb{R}^{2n}} | W^{\sigma} \rho(z)|^2 d z=\mathcal{P} \left[ W^{\sigma} \rho \right]
\end{array}
\label{eqpurity1}
\end{equation}
And thus, we have as before
\begin{equation}
\begin{array}{l l}
\mathcal{P} \left[W^{\omega} \rho \right] =1 & \mbox{if } W^{\omega} \rho \mbox{ is a pure state}\\
& \\
\mathcal{P} \left[W^{\omega} \rho \right] <1 & \mbox{if } W^{\omega} \rho \mbox{ is a mixed state}
\end{array}
\label{eqpurity2}
\end{equation}
We conclude that the main structures of standard quantum mechanics extend trivially to the non-standard symplectic case. However, there are also significative differences concerning the properties of states. We will address these issues in the next section.

\begin{remark}\label{RemarkPPT}
At this point it is worth going back to the partial transpose transformation mentioned in the introduction in relation to the separability problem for systems with continuous variables.

Suppose that a system is constituted of two subsystems - Alice and Bob - with $n^A $ and $n^B$ degrees of freedom $(n=n^A+n^B)$. Their coordinates are $z^A=(x^A, p^A)$ and $z^B=(x^B, p^B)$. We write the collective coordinate $z=(x^A,x^B, p^A, p^B)$. The partial transpose transformation reverses Bob's momentum \cite{Simon,Werner}:
\begin{equation}
z \mapsto {\bf P} z= \left(
\begin{array}{c c c r}
{\bf I} & {\bf 0} & {\bf 0} & {\bf 0} \\
{\bf 0} & {\bf I} & {\bf 0} & {\bf 0} \\
{\bf 0} & {\bf 0} & {\bf I} & {\bf 0} \\
{\bf 0} & {\bf 0} & {\bf 0} & - {\bf I}
\end{array}
\right) ~ \left(
\begin{array}{c}
x^A\\
x^B\\
 p^A\\
  p^B
\end{array}
\right) =
\left(
\begin{array}{c}
x^A\\
x^B\\
 p^A\\
  - p^B
\end{array}
\right)
\label{eqPPT1}
\end{equation}
A straightforward computation reveals that the transformation ${\bf P}={\bf P^{-1}}$ is neither $\sigma$-symplectic nor $\sigma$-anti-symplectic. Consequently, in general, $W^{\sigma} \psi ({\bf P^{-1}}z)$ is not a $\sigma$-Wigner function. Rather, one can think of ${\bf P}$ as a Darboux map for the symplectic form
\begin{equation}
\omega (z, z^{\prime}) = z^T {\bf \Omega^{-1}} z^{\prime},
\label{eqPPT2}
\end{equation}
with
\begin{equation}
{\bf \Omega}= {\bf P} {\bf J} {\bf P^T}.
\label{eqPPT3}
\end{equation}
And hence, $W^{\omega} \psi (z)=W^{\sigma} \psi ({\bf P^{-1}}z)$ is a Wigner function on the non-standard symplectic space with symplectic form (\ref{eqPPT2},\ref{eqPPT3}).
\end{remark}

\section{States in phase-space}

In this section we shall always consider real, normalized and
square-integrable functions defined on the phase-space
$\mathbb{R}^{2n}$. A difficult problem consists in assessing
whether such a function $f$ qualifies as a $\omega$-Wigner
function for some symplectic form $\omega$. In other words, how
can one tell if $f$ is the $\omega$-Wigner function of some
density matrix $\widehat{\rho}$?

It is well known that the answer lies in the following positivity condition \cite{Dias3,Lions}. A real and normalized function $f \in L^2 (\mathbb{R}^{2n})$ is a $\sigma$-Wigner function if and only if
\begin{equation}
\int_{\mathbb{R}^{2n}} f(z) W^{\sigma} \psi (z)dz \ge 0,
\label{eqQuantumConditions1}
\end{equation}
for all $\psi \in L^2 (\mathbb{R}^{n})$. On the other hand, $f$ is a $\omega$-Wigner function if and only if there exist a $\sigma$-Wigner function $f^{\sigma}$ and a matrix ${\bf S} \in \mathcal{D}(n; \omega)$ such that $f(\xi) = f^{\sigma} ({\bf S}^{-1} \xi)$. That happens if and only if
\begin{equation}
\begin{array}{c}
\int_{\mathbb{R}^{2n}} f(\xi) W^{\omega} \psi (\xi)d \xi = \int_{\mathbb{R}^{2n}}  f^{\sigma} ({\bf S}^{-1}\xi) W^{\sigma} \psi ({\bf S}^{-1} \xi)d \xi =\\
\\
= \int_{\mathbb{R}^{2n}} f^{\sigma}(z) W^{\sigma} \psi (z)dz \ge 0,
\end{array}
\label{eqQuantumConditions2}
\end{equation}
for all $\psi \in L^2 (\mathbb{R}^{n})$.

The positivity conditions (\ref{eqQuantumConditions1},\ref{eqQuantumConditions2}) are somewhat tautological, as they require the knowledge of the entire set of $\sigma$- or $\omega$-Wigner functions of pure states. There are an alternative set of conditions which are more in the spirit of measure theory and Bochner's Theorem. They were derived by Kastler \cite{Kastler}, and Loupias and Miracle-Sole \cite{Loupias} using the machinery of $C^{\ast}$-algebras and were later synthesized by Narcowich \cite{Narcowich1} in the framework of the so-called Narcowich-Wigner (NW) spectrum. This proved to be a very powerful tool to address convolutions of functions defined on the phase-space as we shall shortly see.

Here is a brief summary of the remainder of this section. In the next subsection, we define the NW spectrum of a function for a given symplectic structure and show how it provides a condition of positivity alternative to (\ref{eqQuantumConditions1}) or (\ref{eqQuantumConditions2}). In Theorem \ref{TheoremQuantumConditions5}, we show how the NW spectrum behaves under the convolution. In Theorem \ref{TheoremPropertiesNWSpectra} we gather the main properties of NW spectra. In subsection 4.2 we consider the uncertainty principle in the Robertson-Schr\"odinger form and the associated concept of symplectic spectrum for a given symplectic form $\omega$. The main results of this subsection are the Williamson Theorem for a non-standard symplectic form (Theorem \ref{TheoremOmegaWilliamsonTheorem1}) and Theorem \ref{TheoremUnboudedQuotient1} which shows that one can find positive-definite matrices with lowest symplectic eigenvalues with respect to two different sympletic strucutes as far apart as one wishes. All these results culminate in Theorem \ref{TheoremSets1} of section 4.3, where we show how the sets of positive (classical) measures and the sets of Wigner functions for two distinct symplectic structures relate to each other. Finally, as an application, we prove in Theorem \ref{TheoremRepresentabilityProblem} how a Wigner function behaves under an arbitrary linear coordinate transformation.

\subsection{Narcowich-Wigner spectra}

\begin{definition}\label{DefinitionQuantumConditions1}
Let ${\bf \Sigma} \in A(2n; \mathbb{R})$ and $\alpha \in \mathbb{R}$. A complex-valued, continuous function $f$ on $\mathbb{R}^{2n}$ that has the property that the matrix $\left[ \left[M_{jk} \right] \right]$ with entries
\begin{equation}
M_{jk} =f(a_j-a_k) e^{\frac{i \alpha}{2} a_k \cdot {\bf \Sigma} a_j}
\label{eqQuantumConditions3}
\end{equation}
is a positive matrix on $\mathbb{C}^N$ for each $N \in \mathbb{N}$ and all $a_1, a_2 , \cdots, a_N \in \mathbb{R}^{2n}$ is called a function of the $(\alpha, {\bf \Sigma})$-positive type. If $\alpha=0$, then we shall simply say that $f$ is of the $0$-positive type.
\end{definition}

A function may be of the $(\alpha, {\bf \Sigma})$-positive type for various values of $\alpha$ and these can be assembled in a set.

\begin{definition}\label{DefinitionQuantumConditions2}
The ${\bf \Sigma}$-Narcowich-Wigner spectrum of a complex-valued, continuous function $f$ on $\mathbb{R}^{2n}$ is the set:
\begin{equation}
\mathcal{W}^{{\bf \Sigma}} (f) := \left\{ \alpha \in \mathbb{R}: ~ f \mbox{ is of the $(\alpha, {\bf \Sigma})$-positive type} \right\}.
\label{eqQuantumConditions4}
\end{equation}
More generally the Narcowich-Wigner spectrum of $f$ is the set:
\end{definition}
\begin{equation}
\mathcal{W} (f) := \left\{ (\alpha ,{\bf \Sigma}) \in \mathbb{R} \times A (2n; \mathbb{R}) : ~ f \mbox{ is of the $(\alpha, {\bf \Sigma})$-positive type} \right\}.
\label{eqQuantumConditions4.1}
\end{equation}

For each $\alpha \in \mathbb{R}$ and ${\bf \Sigma} \in A(2n; \mathbb{R})$ the set of all functions of $(\alpha, {\bf \Sigma})$-positive type is a convex cone.

Before we proceed let us briefly recall Bochner's Theorem.

\begin{theorem}\label{TheoremQuantumConditions3}
{\bf (Bochner)} The set of Fourier transforms of finite, positive measures on $\mathbb{R}^k$ $(k \in \mathbb{N})$ is exactly the cone of functions of $0$-positive type.
\end{theorem}

The KLM (Kastler, Loupias, Miracle-Sole) conditions are a twisted generalization of Bochner's theorem.

\begin{theorem}\label{TheoremQuantumConditions4}
{\bf (Kastler, Loupias, Miracle-Sole)} Let $f$ be a real function on the phase-space $\mathbb{R}^{2n}$. Then $f$ is a $\sigma$-Wigner function if and only if its Fourier transform $\mathcal{F}(f)$ satisfies the $\sigma$-KLM conditions:

\vspace{0.3 cm}
\noindent
(i) $\mathcal{F}(f) (0)=1$;

\vspace{0.3 cm}
\noindent
(ii) $\mathcal{F}(f)$ is of the $(1,{\bf J})$-positive type.
\end{theorem}

The first condition ensures the normalization of $f$, while the second one is equivalent to the positivity condition (\ref{eqQuantumConditions1}). In terms of the Narcowich-Wigner spectrum, condition (ii) means
\begin{equation}
1  \in \mathcal{W}^{{\bf J}} (\mathcal{F}(f)).
\label{eqQuantumConditions5}
\end{equation}
On the other hand, if $f$ is everywhere non-negative, then according to Bochner's theorem
\begin{equation}
0 \in \mathcal{W}^{{\bf \Sigma}} (\mathcal{F}(f)),
\label{eqQuantumConditions6}
\end{equation}
for all ${\bf \Sigma} \in A(2n; \mathbb{R})$.

We now generalize the KLM conditions (Theorem \ref{TheoremQuantumConditions4}) for an arbitrary symplectic form $\omega$.

\begin{theorem}\label{TheoremQuantumConditions5}
Let $f$ be a real function on the phase-space $\mathbb{R}^{2n}$. Then $f$ is a $\omega$-Wigner function if and only if its Fourier transform $\mathcal{F}(f)$ satisfies the $\omega$-KLM conditions:

\vspace{0.3 cm}
\noindent
(i) $\mathcal{F}(f)(0)=1$;

\vspace{0.3 cm}
\noindent
(ii) $\mathcal{F}(f)$ is of the $(1,{\bf \Omega})$-positive type.
\end{theorem}

\begin{proof}
As usual $f$ is a $\omega$-Wigner function if and only if there exists a $\sigma$-Wigner function $f^{\sigma}$ and a matrix ${\bf S} \in \mathcal{D}(n; \omega)$ such that $f(z)=f^{\sigma}({\bf S}^{-1}z)$. After a simple calculation, we conclude that $\mathcal{F}(f)(a) =\mathcal{F}(f^{\sigma})({\bf S}^T a)$. Consequently:
\begin{equation}
\begin{array}{c}
 \mathcal{F}(f)(a_j-a_k) e^{\frac{i \alpha}{2} a_k \cdot {\bf \Omega} a_j} = \mathcal{F}(f^{\sigma}) ({\bf S}^T(a_j-a_k)) e^{\frac{i \alpha}{2} a_k \cdot {\bf \Omega} a_j}=\\
\\
=\mathcal{F}(f^{\sigma}) (b_j-b_k) e^{\frac{i \alpha}{2} b_k \cdot {\bf S}^{-1} {\bf \Omega} ({\bf S}^T)^{-1}  b_j} =\mathcal{F}(f^{\sigma}) (b_j-b_k) e^{\frac{i \alpha}{2} b_k \cdot  {\bf J}  b_j}
\end{array}
\label{eqQuantumConditions7}
\end{equation}
with $b_j={\bf S}^T a_j$ for $j=1, \cdots, N$. We conclude that $\mathcal{W}^{{\bf \Omega}} (\mathcal{F}(f)) = \mathcal{W}^{{\bf J}} (\mathcal{F}(f^{\sigma}))$ and the result follows.
\end{proof}

\begin{remark}\label{Remark5.1}
It follows from the proof that $(\alpha , {\bf \Omega}) \in \mathcal{W} (\mathcal{F}(f))$ if and only if $(\alpha , {\bf J}) \in \mathcal{W} (\mathcal{F}(f^{\sigma}))$.
\end{remark}

Before we proceed, let us recapitulate the concept of Hamard-Schur (HS) product. Let ${\bf A}=\left[ \left[ A_{ij} \right] \right], {\bf B} =\left[ \left[ B_{ij} \right] \right] \in M (k \times n; \mathbb{C})$. Then the HS product is the map $\diamondsuit:  M (k \times n; \mathbb{C}) \times  M (k \times n; \mathbb{C}) \to  M (k \times n; \mathbb{C})$
\begin{equation}
{\bf A} \diamondsuit {\bf B} = \left[ \left[ A_{ij} B_{ij} \right] \right].
\label{eqQuantumConditions8}
\end{equation}
According to Schur's Theorem, the HS product has the property of being a closed operation in the convex cone $\mathbb{P}_{k \times k} (\mathbb{C})$ of positive $k \times k$ matrices, i.e.:
\begin{equation}
\diamondsuit: \mathbb{P}_{k \times k} (\mathbb{C}) \times \mathbb{P}_{k \times k} (\mathbb{C}) \to \mathbb{P}_{k \times k} (\mathbb{C}).
\label{eqQuantumConditions9}
\end{equation}

\begin{theorem}\label{TheoremQuantumConditions5}
Let $f,g \in L^1 (\mathbb{R}^{2n})$ be real-valued functions on
$\mathbb{R}^{2n}$, such that $\mathcal{F}(f),\mathcal{F}(g)$ are
of the $(\alpha, {\bf \Sigma})$ and of the $(\beta, {\bf
\Upsilon})$-positive type, respectively, and $\det(\alpha {\bf
\Sigma} + \beta  {\bf \Upsilon}) \ne 0$. Moreover, define $\gamma
:= \left(\det(\alpha {\bf \Sigma} + \beta  {\bf \Upsilon})
\right)^{\frac{1}{2n}}$ to be a real $(2n)$-th root of that
determinant (this is always possible since real antisymmetric
matrices have non-negative determinants; they are equal to the
square of a polynomial of its entries - the Pfaffian). Then the
Fourier transform of the convolution $f \star g$ is of the
$\left(\gamma , \frac{\alpha}{\gamma} {\bf \Sigma} +
\frac{\beta}{\gamma} {\bf \Upsilon}\right)$-positive type. For
later convenience we write:
\begin{equation}
(\alpha, {\bf \Sigma}) \oplus (\beta, {\bf \Upsilon}) = \left(\gamma , \frac{\alpha}{\gamma} {\bf \Sigma} + \frac{\beta}{\gamma} {\bf \Upsilon}\right).
\label{eqQuantumConditions9.1}
\end{equation}
\end{theorem}

\begin{proof}
Since  $f,g \in L^1 (\mathbb{R}^{2n})$, the Fourier transform of the convolution amounts to the pointwise product $\mathcal{F}(f \star g) = \mathcal{F}(f) \cdot \mathcal{F}(g) $. Thus:
\begin{equation}
\begin{array}{c}
\mathcal{F}(f \star g) (a_j -a_k) e^{\frac{i \gamma}{2} \left( \frac{\alpha}{\gamma} a_k \cdot {\bf \Sigma} + \frac{\beta}{\gamma} {\bf \Upsilon}\right) a_j } =\\
\\
= \left(\mathcal{F}(f) (a_j -a_k) e^{\frac{i \alpha}{2} a_k \cdot {\bf \Sigma} a_j} \right) \left(\mathcal{F}(g) (a_j -a_k) e^{\frac{i \beta}{2} a_k \cdot {\bf \Upsilon} a_j} \right).
\end{array}
\label{eqQuantumConditions10}
\end{equation}
The last expression is the $(jk)$-th component of the HS product of two positive matrices and the result follows.
\end{proof}

Henceforth, we shall use the notation $\mathcal{G}_{{\bf A}} $ to represent the phase-space Gaussian with covariance matrix ${\bf A}$:
\begin{equation}
\mathcal{G}_{{\bf A}} (z) = \frac{1}{(2 \pi)^n \sqrt{\det {\bf A}}} \exp \left(- \frac{1}{2} z \cdot {\bf A^{-1}} z \right).
\label{eqGaussian1}
\end{equation}
We are tacitly assuming that the expectation values $<\widehat{Z}_i>$ are all equal to zero, something which can easily be achieved by a phase-space translation.

For later convenience we assemble the following properties of Narcowich-Wigner spectra (see \cite{Brocker,Dias4,Narcowich1} for the standard symplectic case):
\begin{theorem}\label{TheoremPropertiesNWSpectra}
Let $\Omega \in A (2n; \mathbb{R})$, $\mathcal{W}^{{\bf \Omega}} (\mathcal{F}(f))$ denote ${\Omega}$-Narcowich-Wigner spectrum of $\mathcal{F}(f)$ for some function $f \in L^2 (\mathbb{R}^{2n})$ with continuous Fourier transform. Then the following properties hold.

\begin{enumerate}

\item If $\mathcal{F}(f)$ is of $(\alpha,{\bf \Omega})$-positive type for some $\alpha \in \mathbb{R}$ and some ${\bf \Omega}\in A(2n; \mathbb{R})$, then $f$ is a real function.

\item $ \mathcal{W}^{{\bf \Omega}} (\mathcal{F}(f^{\vee}))  =\mathcal{W}^{{\bf \Omega}} (\mathcal{F}(f))$, where $f^{\vee } (z)=f(-z)$.

\item $\alpha \in \mathcal{W}^{{\bf \Omega}} (\mathcal{F}(f))$ if and only if $-\alpha  \in \mathcal{W}^{{\bf \Omega}} (\mathcal{F}(f))$.

\item Let $f_{\lambda}(z)= | \lambda|^{2n} f(\lambda z)$ for $\lambda \in \mathbb{R} \backslash \left\{0 \right\}$. Then $\mathcal{W}^{{\bf \Omega}} (\mathcal{F}(f_{\lambda})) = \frac{1}{\lambda^2} \mathcal{W}^{{\bf \Omega}} (\mathcal{F}(f))$.

\item Let $\det {\bf \Omega}=1$ and $\omega (z,z^{\prime}) = z \cdot    {\bf \Omega^{-1}} z^{\prime}$ is the associated symplectic form. For ${\bf M} \in Sl(2n; \mathbb{R})$ define $f_{{\bf M}} (z)=f({\bf M}z)$. Then if ${\bf M}$ is $\omega$-symplectic or $\omega$-anti symplectic, we have $\mathcal{W}^{{\bf \Omega}} (\mathcal{F}(f_{{\bf M}}))=\mathcal{W}^{{\bf \Omega}} (\mathcal{F}(f))$.

\item Let $\psi\in L^2 (\mathbb{R}^n)$ be some pure state and $\omega(z,z^{\prime})=z^{\prime}\cdot {\bf \Omega^{-1}} z$ an arbitrary symplectic form. If $\psi$ is a Gaussian function, then $    \mathcal{W}^{{\bf \Omega}} (\mathcal{F}(W^{\omega} \psi)) = \left[-1,1 \right]$. Otherwise $    \mathcal{W}^{{\bf \Omega}} (\mathcal{F}(W^{\omega} \psi)) = \left\{-1,1 \right\}$.
\end{enumerate}
\end{theorem}

\begin{proof}
\begin{enumerate}
\item Suppose that $(\alpha,{\bf \Omega}) \in \mathcal{W}(\mathcal{F}(f))$. Then the matrices with entries $M_{jk}=\mathcal{F}(f)(a_j-a_k) e^{\frac{i \alpha}{2} a_k \cdot {\bf \Omega} a_j}$ are positive in $\mathbb{C}^N$. For $N=2$, we have that
    $$
    \left(\begin{array}{c c}
    \mathcal{F}(f)(0) & \mathcal{F}(f) (a_1-a_2)e^{\frac{i \alpha}{2} a_2 \cdot {\bf \Omega} a_1}\\
    & \\
    \mathcal{F}(f) (a_2-a_1)e^{\frac{i \alpha}{2} a_1 \cdot {\bf \Omega} a_2} & \mathcal{F}(f)(0)
    \end{array} \right)
    $$
is positive. In particular it has to be self-adjoint, and setting $\omega:=a_1-a_2$, we obtain:
$$
\overline{\mathcal{F}(f) (\omega)}=  \mathcal{F}(f) (-\omega)
$$
for all $\omega \in \mathbb{R}^{2n}$, where we used the continuity of $\mathcal{F}(f)$. This is equivalent to $f$ being a real function.

\item A simple calculation shows that $\mathcal{F}(f^{\vee})=\left[\mathcal{F}(f) \right]^{\vee}$. It follows that
    $$
    \mathcal{F}(f^{\vee}) (a_j-a_k)e^{\frac{i \alpha}{2} a_k \cdot {\bf \Omega} a_j}=\mathcal{F}(f) (a_k-a_j)e^{\frac{i \alpha}{2} a_k \cdot {\bf \Omega} a_j}=\mathcal{F}(f) (b_j-b_k)e^{\frac{i \alpha}{2} b_k \cdot {\bf \Omega} b_j}
$$
where $b_j=-a_j$, etc. If follows that $(\alpha , {\bf \Omega}) \in \mathcal{W}^{{\bf \Omega}} (\mathcal{F}(f^{\vee}))$ if and only if $(\alpha , {\bf \Omega}) \in \mathcal{W}^{{\bf \Omega}} (\mathcal{F}(f))$.

\item We have
$$
\begin{array}{c}
\mathcal{F}(f) (a_j-a_k)e^{-\frac{i \alpha}{2} a_k \cdot {\bf \Omega} a_j} =\left[\mathcal{F}(f) \right]^{\vee} (a_k-a_j)e^{\frac{i \alpha}{2} a_j \cdot {\bf \Omega} a_k}=\\
 \\
 =\mathcal{F}(f^{\vee}) (a_k-a_j)e^{\frac{i \alpha}{2} a_j \cdot {\bf \Omega} a_k}
\end{array}
$$
and the result follows from 2.

\item Since $\mathcal{F} f_{\lambda} (z) = \mathcal{F} f \left(\frac{z}{\lambda} \right)$, we have that
$$
\mathcal{F} f_{\lambda} (a_j-a_k)e^{-\frac{i \alpha}{2} a_k \cdot {\bf \Omega} a_j} = \mathcal{F} f \left(\frac{a_j}{\lambda} - \frac{a_k}{\lambda}\right) e^{-\frac{i \alpha}{2} a_k \cdot {\bf \Omega} a_j}=\mathcal{F} f (b_j-b_k)e^{-\frac{i \alpha \lambda^2}{2} b_k \cdot {\bf \Omega} b_j},
$$
with $a_j= \lambda b_j$. We conclude that $\alpha \in \mathcal{W}^{{\bf \Omega}} (\mathcal{F}(f_{\lambda}))$ if and only if $\alpha \lambda^2 \in \mathcal{W}^{{\bf \Omega}} (\mathcal{F}(f))$.

\item We have that ${\bf M}{\bf \Omega} {\bf M^T} = \epsilon {\bf \Omega}$ with $\epsilon=+1$ (resp. $\epsilon=-1$) if ${\bf M}$ is $\omega$-symplectic (resp. $\omega$-anti symplectic).

  On the other hand, we have $\mathcal{F} f_{{\bf M}} (z)= \mathcal{F} f (({\bf M^{-1}})^T z)$. Thus
 $$
 \begin{array}{c}
\mathcal{F} f_{{\bf M}} (a_j-a_k)e^{-\frac{i \alpha}{2} a_k \cdot {\bf \Omega} a_j} = \mathcal{F} f \left(({\bf M^{-1}})^T a_j - ({\bf M^{-1}})^T a_k \right) e^{-\frac{i \alpha}{2} a_k \cdot {\bf \Omega} a_j}=\\
\\
=\mathcal{F} f (b_j-b_k)e^{-\frac{i \alpha }{2} b_k \cdot {\bf M} {\bf \Omega} {\bf M^T }b_j}=\mathcal{F} f (b_j-b_k)e^{-\frac{i \alpha \epsilon }{2} b_k \cdot  {\bf \Omega} b_j},
\end{array}
$$
where $a_j={\bf M^T }b_j$. We conclude that $\alpha \in \mathcal{W}^{{\bf \Omega}} (\mathcal{F}(f_{{\bf M}})$ if and only if $\epsilon \alpha  \in \mathcal{W}^{{\bf \Omega}} (\mathcal{F}(f))$. The rest is a consequence of 3.

\item We proved in \cite{Dias4} that $    \mathcal{W}^{{\bf \sigma}} (\mathcal{F}(W^{\sigma} \psi)) = \left[-1,1 \right]$ if $\psi$ is a Gaussian function, and $    \mathcal{W}^{{\bf \sigma}} (\mathcal{F}(W^{\sigma} \psi)) = \left\{-1,1 \right\}$ otherwise. The rest is a consequence of Remark \ref{Remark5.1}.
\end{enumerate}
\end{proof}

\subsection{Uncertainty principle and symplectic spectra}

Consider again the Gaussian measure (\ref{eqGaussian1}). It is well known that this is the Wigner function associated with some density matrix if and only if \cite{Narcowich1}:
\begin{equation}
{\bf A}+\frac{i}{2} {\bf J} \ge 0,
\label{eqGaussian2}
\end{equation}
that is ${\bf A}+\frac{i}{2} {\bf J} $ is a positive matrix in $\mathbb{C}^{2n}$. This matrix inequality is known in the literature as the Robertson-Schr\"odinger uncertainty principle (RSUP). It is stronger than the more familiar Heisenberg inequality (see below), as it also accounts for the position-momentum correlations. It also has the advantage of being invariant under linear symplectic transformations. Indeed, if $W^{\sigma} \rho (z) $ is the Wigner function of some density matrix $\widehat{\rho}$ and ${\bf P} \in Sp (n; \sigma)$, then $W^{\sigma} \rho ({\bf P} z) $ is the Wigner function $W^{\sigma} \rho_{{\bf P}}$ of some other density matrix $\widehat{\rho}_{{\bf P}}$, related to $\widehat{\rho}$ by a metaplectic transformation \cite{Leray,Shale,Weil}. Accordingly, if ${\bf A_{\rho}}$ and ${\bf A_{\rho_{P}}}$ are the corresponding covariance matrices, then the two are related by the following transformation:
\begin{equation}
{\bf A_{\rho}} = {\bf P} {\bf A_{\rho_{P}}} {\bf P^T},
\label{eqGaussian3.1}
\end{equation}
where we used the fact that $\det {\bf P}=1$ for ${\bf P} \in Sp(n; \sigma)$. Thus:
\begin{equation}
{\bf A_{\rho}} +\frac{i}{2} {\bf J} \ge 0 \Leftrightarrow {\bf A_{\rho_{P}}}   +\frac{i}{2} {\bf P^{-1} } {\bf J} {\bf\left( P^{-1} \right)^T}  \ge 0 \Leftrightarrow {\bf A_{\rho_{P}}}   +\frac{i}{2} {\bf J} \ge 0,
\label{eqGaussian3}
\end{equation}
since ${\bf P^{-1} }  \in Sp(n; \sigma)$.

To establish whether a state satisfies the RSUP (\ref{eqGaussian2}), one computes the eigenvalues of the matrix ${\bf A}+\frac{i}{2} {\bf J} $ and verifies if they are all greater or equal to zero. Alternatively, we may verify the RSUP by resorting to the symplectic eigenvalues of ${\bf A}$ and Williamson's Theorem.

First notice that if ${\bf A} \in \mathbb{P}^{\times}_{2n \times 2n} (\mathbb{R})$, then the matrix ${\bf A} {\bf J^{-1}}$ has the same eigenvalues as ${\bf A^{1/2}} {\bf J^{-1}} {\bf A^{1/2}}$ \cite{Williamson}. So the eigenvalues of ${\bf A} {\bf J^{-1}}$ come in pairs $\pm i \lambda_{\sigma, j}({\bf A})$, where $\lambda_{\sigma, j}({\bf A}) >0$, $j=1, \cdots, n$.

\begin{definition}\label{DefinitionWilliamsontheorem1}
Let $\lambda_{\sigma, j}({\bf A})$, $j=1, \cdots, n$ denote the moduli of the eigenvalues of ${\bf A} {\bf J^{-1}}$ written as an increasing sequence
\begin{equation}
0  < \lambda_{\sigma, 1}({\bf A}) \le \lambda_{\sigma, 2}({\bf A}) \le \cdots \le \lambda_{\sigma, n}({\bf A}).
\label{eqGaussian3}
\end{equation}
They are called the $\sigma$-Williamson invariants or $\sigma$-symplectic eigenvalues of ${\bf A}$, and the $n$-tuple
\begin{equation}
Spec_{\sigma} ({\bf A}):= \left( \lambda_{\sigma, 1}({\bf A}) , \lambda_{\sigma, 2}({\bf A}) , \cdots , \lambda_{\sigma, n}({\bf A}) \right)
\label{eqGaussian4}
\end{equation}
is called the $\sigma$-symplectic spectrum of ${\bf A}$.
\end{definition}

The invariance of the $\sigma$-symplectic spectrum of ${\bf A}$ under the transformation ${\bf A} \mapsto {\bf A^{\prime}} ={\bf P} {\bf A} {\bf P^T}$ with ${\bf P} \in Sp(n; \sigma)$ is easily established. Indeed suppose that $\lambda \in Spec_{\sigma} ({\bf A})$. Then
\begin{equation}
\begin{array}{c}
0 = \det ({\bf A}{\bf J^{-1}} \pm i \lambda {\bf I}) = \det ({\bf P^{-1}}{\bf A^{\prime}}{\bf P^{-1 T}}{\bf J^{-1}} \pm i \lambda {\bf I}) =\\
 \\
 = \det ({\bf P^{-1}}{\bf A^{\prime}}{\bf J^{-1}}{\bf P} \pm i \lambda {\bf I}) = \det ({\bf A^{\prime}}{\bf J^{-1}} \pm i \lambda {\bf I}),
\end{array}
\label{eqGaussian5}
\end{equation}
where we used $ {\bf J} {\bf P^T} = {\bf P^{-1}} {\bf J}$. Thus $Spec_{\sigma} ({\bf A^{\prime}})=Spec_{\sigma} ({\bf A})$.

The following theorem is due to Williamson \cite{Williamson}:

\begin{theorem}\label{Williamsontheorem2}
{\bf (Williamson)} Let $A \in \mathbb{P}^{\times}_{2n \times 2n} (\mathbb{R})$. There exists a symplectic matrix ${\bf P} \in Sp(n; \sigma)$ such that
\begin{equation}
{\bf P}  {\bf A} {\bf P^T} = diag \left(\lambda_{\sigma, 1}({\bf A}) ,  \cdots , \lambda_{\sigma, n}({\bf A}), \lambda_{\sigma, 1}({\bf A}) ,  \cdots , \lambda_{\sigma, n}({\bf A}) \right) ,
\label{eqGaussian6}
\end{equation}
where $\lambda_{\sigma, j}({\bf A})$, $j=1, \cdots, n$ are the $\sigma$-Williamson invariants of ${\bf A}$.
\end{theorem}

We are now in a position to reexpress the RSUP in terms of the $\sigma$-symplectic spectrum. Let ${\bf A}$ be the covariance matrix of some Wigner function $W^{\sigma} \rho (z)$ and let ${\bf D}$ denote the diagonal matrix appearing in (\ref{eqGaussian6}). We thus have
\begin{equation}
{\bf A} + \frac{i}{2} {\bf J} \ge 0 \Leftrightarrow {\bf D} + \frac{i}{2} {\bf P} {\bf J} {\bf P^T} \ge 0 \Leftrightarrow {\bf D} + \frac{i}{2} {\bf J} \ge 0 .
\label{eqGaussian7}
\end{equation}
The eigenvalues of the last matrix are easily computed. They are equal to $\lambda_{\sigma, j} ({\bf A}) \pm \frac{1}{2}$, $j=1, \cdots, n$. Thus the last inequality in (\ref{eqGaussian7}) holds if and only if \cite{Narcowich1}
\begin{equation}
\lambda_{\sigma, 1} ({\bf A}) \ge \frac{1}{2}.
\label{eqGaussian8}
\end{equation}
Let now $W^{\sigma} \rho (z)$ be some Wigner function with covariance matrix ${\bf A}$. Let ${\bf P} \in Sp(n; \sigma)$ be the symplectic matrix diagonalizing ${\bf A}$ as in (\ref{eqGaussian6}). As we argued before,  $W^{\sigma} \rho ({\bf P^{-1}}z)$ corresponds to another Wigner function $W^{\sigma} \rho_{{\bf P}} (z)$ with covariance matrix ${\bf P} {\bf A} {\bf P^T} = {\bf D}$ (cf.(\ref{eqGaussian3.1},\ref{eqGaussian6})). This means that since the covariance matrix of $W^{\sigma} \rho_{{\bf P}} (z)$ is diagonal, we have
\begin{equation}
<\widehat{q}_j^2>_{\rho_{{\bf P}}} = <\widehat{p}_j^2>_{\rho_{{\bf P}}} = \lambda_{\sigma, j} ({\bf A}), \hspace{1 cm} j=1, \cdots, n
\label{eqGaussian9}
\end{equation}
Thus
\begin{equation}
<\widehat{q}_j^2>_{\rho_{{\bf P}}} <\widehat{p}_j^2>_{\rho_{{\bf P}}} = \left(\lambda_{\sigma, j} ({\bf A}) \right)^2 \ge \frac{1}{4} , \hspace{1 cm} j=1, \cdots, n
\label{eqGaussian10}
\end{equation}
which is the Heisenberg uncertainty principle. Consequently, the lowest value of $\frac{1}{2}$ for each $\sigma$-symplectic eigenvalue (cf.(\ref{eqGaussian8})), corresponds to a minimal uncertainty in some direction in the phase-space. Thus the extremal situation, where
\begin{equation}
\lambda_{\sigma, 1} ({\bf A}) =\lambda_{\sigma, 2} ({\bf A})  = \cdots = \lambda_{\sigma, n} ({\bf A})=\frac{1}{2},
\label{eqGaussian11}
\end{equation}
(which is equivalent to $2 {\bf A} \in Sp(n; \sigma)$) can only be achieved by a Gaussian pure state. This is known in the literature as Littlejohn's Theorem \cite{Bastiaans,Littlejohn}.

For future reference we state and prove the following proposition.

\begin{proposition}\label{PropostionEigenvector}
Let $e, f \in \mathbb{R}^{2n}$ be such that $\sigma (f,e) \ne 0$.
Then there exists a matrix ${\bf A} \in \mathbb{P}^{\times}_{2n
\times 2n} ( \mathbb{R})$ such that $u=e + if$ is an eigenvector
of $i {\bf A} {\bf J^{-1}}$ with eigenvalue equal to
$\lambda_{\sigma,1} ({\bf A})$ (if $\sigma (f,e)>0$) or $-
\lambda_{\sigma,1} ({\bf A})$ (if $\sigma (f,e)<0$), where
$\lambda_{\sigma,1} ({\bf A})$ is the smallest $\sigma$-Williamson
invariant of ${\bf A}$.
\end{proposition}

\begin{proof}
Suppose that $\lambda_1:= \sigma (f,e) >0$. Set $f^{\prime} =
\frac{f}{\sqrt{\lambda_1}}$ and $e^{\prime} =
\frac{e}{\sqrt{\lambda_1}}$. It follows that $\sigma (f^{\prime},
e^{\prime}) =1$. A well known theorem in symplectic geometry
\cite{Cannas,Gosson1} states that we can find a symplectic basis
$\left\{e_i^{\prime},f_j^{\prime} \right\}_{1\le i,j \le n}$ such
that $e_1^{\prime}= e^{\prime}$ and $f_1^{\prime}=f^{\prime}$.
Next choose an arbitrary set of positive numbers
$\lambda_2,\lambda_3, \cdots , \lambda_n$ such that $\lambda_1 \le
\lambda_2 \le  \cdots \le \lambda_n$, and define
\begin{equation}
e_i = \sqrt{\lambda_i} e_i^{\prime}, \hspace{1 cm} f_j = \sqrt{\lambda_j} f_j^{\prime}, \hspace{1 cm} 1 \le i,j \le n.
\label{eqEigenvector1}
\end{equation}
Now define $v_i=e_i$, and $v_{i+n}=f_i$ (and similarly $v_i^{\prime}=e_i^{\prime}$, and $v_{i+n}^{\prime}=f_i^{\prime}$), and $\lambda_{i+n}= \lambda_i$ for $i=1, \cdots, n$. We can now rewrite (\ref{eqEigenvector1}) in a more compact manner:
\begin{equation}
v_{\alpha} = \sqrt{\lambda_{\alpha}} v_{\alpha}^{\prime}, \hspace{1 cm} 1 \le \alpha \le 2n .
\label{eqEigenvector2}
\end{equation}
Since $\left\{v_{\alpha}^{\prime} \right\}_{1 \le \alpha \le 2n} $ is a symplectic basis, we have:
\begin{equation}
\sigma (v_{\alpha}, v_{\beta})=  \sqrt{\lambda_{\alpha} \lambda_{\beta}} \sigma (v_{\alpha}^{\prime}, v_{\beta}^{\prime}) =  \sqrt{\lambda_{\alpha} \lambda_{\beta}} J_{\beta \alpha} = \lambda_{\alpha} J_{\beta \alpha},
\label{eqEigenvector3}
\end{equation}
where we used the fact that $\lambda_{i+n}= \lambda_i$, for $i=1,2, \cdots, n$. In general the basis $\left\{v_{\alpha} \right\}_{1 \le \alpha \le 2n} $ is not orthonormal. However, there exists a matrix ${\bf A} \in \mathbb{P}^{\times}_{2n \times 2n} ( \mathbb{R})$ such that
\begin{equation}
v_{\alpha}^T {\bf A^{-1}} v_{\beta} = \delta_{\alpha \beta}, \hspace{1 cm} 1 \le \alpha, \beta \le 2n.
\label{eqEigenvector4}
\end{equation}
The matrix ${\bf A}$ defines an inner product in $\mathbb{C}^{2n}$:
\begin{equation}
<z, z^{\prime}>_{{\bf A}} := \overline{z} \cdot {\bf A^{-1}} z^{\prime}, \hspace{1 cm} z, z^{\prime} \in \mathbb{C}^{2n}.
\label{eqEigenvector5}
\end{equation}

We conclude the proof by showing that $(\lambda_1, \lambda_2, \cdots, \lambda_n)$ is the symplectic spectrum of ${\bf A}$ and that ${\bf A} {\bf J^{-1}}(e +i f) = -i \lambda_1 (e +i f)$.

Let $T : \mathbb{R}^{2n}\to \mathbb{R}^{2n}$ be the linear transformation $T(v)={\bf AJ^{-1}} v$, with matrix representation ${\bf AJ^{-1}}$ in the canonical basis. Then
\begin{equation}
<v_{\alpha}, T\left( v_{\beta} \right)>_{{\bf A}} = \sigma (v_{\alpha}, v_{\beta})=  \lambda_{\alpha} J_{\beta \alpha}.
\label{eqEigenvector6}
\end{equation}
So the matrix representation ${\bf T}$ of $T$ with respect to the basis $\left\{v_{\alpha} \right\}_{1 \le \alpha \le 2n}$ is
\begin{equation}
{\bf T}= \left(
\begin{array}{c c}
{\bf 0} & - {\bf \Lambda}\\
 {\bf \Lambda} & {\bf 0}
\end{array}
\right),
\label{eqEigenvector7}
\end{equation}
where ${\bf \Lambda} = diag (\lambda_1, \cdots, \lambda_n)$. The eigenvalues of ${\bf T}$ (and hence of $T$) are easily computed. They are: $\left\{\pm i \lambda_j\right\}_{j=1, \cdots,n}$ and the associated eigenvectors of ${\bf AJ^{-1}}$ are:
\begin{equation}
{\bf AJ^{-1}} (e_j \pm i f_j) = \mp i \lambda_j (e_j \pm i f_j), \hspace{1 cm} j=1, \cdots, n.
\label{eqEigenvector8}
\end{equation}
In particular ${\bf AJ^{-1}} (e+if) = - i \lambda_1 (e+if)$.

If $\sigma (f,e)<0$, then by following an identical procedure, we would obtain: ${\bf AJ^{-1}} (e+if) =  i \lambda_1 (e+if)$.
\end{proof}

Now we turn to the non-standard symplectic space $(\mathbb{R}^{2n}; \omega)$. Suppose that $W^{\omega} \rho (z)$ is some $\omega$-Wigner function and let ${\bf S} \in \mathcal{D} (n; \omega)$ denote an arbitrary Darboux matrix. According to (\ref{eqNonstandardWeylalgebra13})
\begin{equation}
W^{\sigma} \rho (z) :=  W^{\omega} \rho ({\bf S} z)
\label{eqGaussian12}
\end{equation}
is a $\sigma$-Wigner function. Hence it has to comply with RSUP. Let ${\bf A}^{\sigma}$ and ${\bf A}^{\omega}$ denote the covariance matrices of $W^{\sigma} \rho $ and $W^{\omega} \rho $, respectively. We thus have
\begin{equation}
{\bf A}^{\omega} = {\bf S} {\bf A}^{\sigma} {\bf S^T}.
\label{eqGaussian13}
\end{equation}
Consequently,
\begin{equation}
{\bf A}^{\sigma} + \frac{i}{2} {\bf J} \ge 0 \Leftrightarrow {\bf A}^{\omega} + \frac{i}{2} {\bf S} {\bf J} {\bf S^T }\ge 0 \Leftrightarrow  {\bf A}^{\omega} + \frac{i}{2} {\bf \Omega} \ge 0.
\label{eqGaussian14}
\end{equation}
The last inequality is valid irrespective of whatever Darboux matrix ${\bf S}$ we use. We shall call this inequality the $\omega$-RSUP.

Next let ${\bf A} \in \mathbb{P}^{\times}_{2n \times 2n} (\mathbb{R})$, and consider the eigenvalues of the matrix ${\bf A} {\bf \Omega^{-1}}$. As in the standard case ${\bf A J^{-1}}$, these come in pairs $\pm i \lambda_{\omega, j} ({\bf A})$, $\lambda_{\omega, j} ({\bf A}) >0$, $j=1, \cdots, n$.

\begin{definition}\label{DefinitionOmegaSymplecticSpectrum1}
The moduli of the eigenvalues of ${\bf A} {\bf \Omega^{-1}}$ written as an increasing sequence $0 < \lambda_{\omega, 1} ({\bf A}) \le \lambda_{\omega, 2} ({\bf A}) \le \cdots \le \lambda_{\omega, n} ({\bf A})$ will be called the $\omega$-Williamson invariants or the $\omega$-symplectic eigenvalues of ${\bf A}$. Moreover, the $n$-tuple
\begin{equation}
Spec_{\omega} ({\bf A}) := \left( \lambda_{\omega, 1} ({\bf A}) , \lambda_{\omega, 2} ({\bf A}) , \cdots , \lambda_{\omega, n} ({\bf A}) \right)
\label{eqGaussian15}
\end{equation}
is called the $\omega$-symplectic spectrum of ${\bf A}$.
\end{definition}
It is easy to show that the $\omega$-symplectic spectrum of ${\bf A}$ is invariant under the $\omega$-symplectic transformation ${\bf A} \mapsto {\bf P} {\bf A} {\bf P^T}$, with ${\bf P} \in Sp (n; \omega)$. The proof follows the same steps as in (\ref{eqGaussian5}).

\begin{lemma}\label{LemmaRescalings1}
Let ${\bf A} \in \mathbb{P}^{\times}_{2n \times 2n} (\mathbb{R})$ and $\mu >0$. The matrix ${\bf B} = \mu {\bf A}$ has $\omega$-symplectic spectrum $Spec_{\omega} ({\bf B}) = \mu Spec_{\omega} ({\bf A})$.
\end{lemma}

\begin{proof}
This is a simple consequence of the following identity:
\begin{equation}
\det ({\bf B} {\bf \Omega}^{-1} - \lambda {\bf I} ) = \mu^{2n} \det \left({\bf A} {\bf \Omega}^{-1} - \frac{\lambda}{\mu} {\bf I} \right),
\label{eqLemmaRescalings1}
\end{equation}
and thus $\lambda \in Spec_{\omega} ({\bf B})$ if and only if $\frac{\lambda}{\mu} \in Spec_{\omega} ({\bf A})$.
\end{proof}

Next, we state and prove the counterpart of Williamson's Theorem on $(\mathbb{R}^{2n}; \omega)$.

\begin{theorem}\label{TheoremOmegaWilliamsonTheorem1}
Let $A \in \mathbb{P}^{\times}_{2n \times 2n} (\mathbb{R})$. Then there exists
a Darboux matrix ${\bf S} \in \mathcal{D}(n; \omega)$ such that
\begin{equation}
{\bf S^{-1}} {\bf A}{\bf  \left(S^{-1}\right)^T} = diag \left(\lambda_{\omega,1} ({\bf A}) , \cdots, \lambda_{\omega,n} ({\bf A})
, \lambda_{\omega,1} ({\bf A}) , \cdots, \lambda_{\omega,n} ({\bf A})\right),
\label{eq12.4}
\end{equation}
where $\left\{\lambda_{\omega,j} ({\bf A})\right\}_{1 \le j \le n}$ are the $\omega$-Williamson invariants of ${\bf A}$.
\end{theorem}

\begin{proof} Let ${\bf S^{\prime}}  \in \mathcal{D} (n; \omega)$ be a Darboux matrix. Then ${\bf S^{\prime -1}} {\bf A} ({\bf  S^{\prime -1}})^T \in \mathbb{P}^{\times}_{2n \times 2n} (\mathbb{R})$. By the $\sigma$-Williamson Theorem, there exists ${\bf M} \in Sp(n; \sigma)$ such that
\begin{equation}
{\bf M^T} {\bf S^{\prime -1}} {\bf A} ({\bf S^{\prime -1}})^T  {\bf M} = {\bf D},
\label{eq12.5}
\end{equation}
where
\begin{equation}
{\bf D}=diag \left(\lambda_{\sigma,1} ({\bf B}) , \cdots, \lambda_{\sigma,n} ({\bf B})
,  \lambda_{\sigma,1} ({\bf B}) , \cdots, \lambda_{\sigma,n} ({\bf B})\right)
\label{eq12.6}
\end{equation}
and $\left(\lambda_{\sigma,1} ({\bf B}) , \cdots, \lambda_{\sigma,n} ({\bf B}) \right)$ is the $\sigma$-symplectic spectrum of ${\bf B} = {\bf S^{\prime -1}} {\bf A} ({\bf S^{\prime -1}})^T $.

Let ${\bf S}= {\bf S^{\prime}} ({\bf M^T})^{-1}$. By Lemma \ref{LemmaDarbouxmap1}, ${\bf S} \in \mathcal{D} (n; \omega)$ is a Darboux matrix. The result then follows if we prove that $\left(\lambda_{\sigma,1} ({\bf B}) , \cdots, \lambda_{\sigma,n} ({\bf B}) \right)$ is the $\omega$-symplectic spectrum of ${\bf A}$. We have that $\lambda \in Spec_{\omega}({\bf A})$ if and only if
\begin{equation}
\begin{array}{c}
0=\det({\bf A} {\bf \Omega^{-1}} \pm i  \lambda {\bf I}) =  \det ({\bf S} {\bf D} {\bf S^T} \pm i \lambda {\bf \Omega}) = \\
\\
=\det ({\bf D}\pm i \lambda {\bf S^{-1}} {\bf \Omega} ({\bf S^T})^{-1} ) = \det ({\bf D} \pm i \lambda {\bf J}),
\end{array}
\label{eq12.7}
\end{equation}
that is if and only if $\lambda \in Spec_{\sigma}({\bf B})$.
\end{proof}

There is a converse to the previous result.

\begin{theorem}\label{theorem2}
Let ${\bf A} \in \mathbb{P}^{\times}_{2n \times 2n} (\mathbb{R})$ and consider a set of positive numbers $0 < \lambda_1 \le \lambda_2 \le \cdots \le \lambda_n$. Then there exists a matrix ${\bf P} \in Gl (2n)$ such that
\begin{equation}
{\bf P^{-1}} {\bf A} {\bf \left(P^{-1}\right)^T} = diag (\lambda_1 , \cdots , \lambda_n, \lambda_1 , \cdots , \lambda_n).
\label{eqTheo15}
\end{equation}
Moreover, the set $Spec_{\delta} ({\bf A})= (\lambda_1 , \cdots , \lambda_n)$ is the $\delta$-symplectic spectrum with respect to the symplectic form
\begin{equation}
\delta (z, z^{\prime}) = z \cdot {\bf \Delta^{-1}} z,
\label{eqTheo16}
\end{equation}
with
\begin{equation}
{\bf \Delta} = {\bf P} {\bf J} {\bf P^T}.
\label{eqTheo17}
\end{equation}
\end{theorem}

\begin{proof}
Let $Spec_{\sigma} ({\bf A})=(\lambda_{\sigma,1} ({\bf A}) , \cdots , \lambda_{\sigma,n} ({\bf A}))$ denote the $\sigma$-symplectic spectrum of ${\bf A}$ with respect to the standard symplectic form $\sigma$. By Williamson's Theorem, there exists a $\sigma$-symplectic matrix ${\bf M} \in Sp(n; \sigma)$ such that
\begin{equation}
{\bf M^T} {\bf A} {\bf M} = diag (\lambda_{\sigma,1} ({\bf A}) , \cdots , \lambda_{\sigma,n} ({\bf A}),\lambda_{\sigma,1} ({\bf A}) , \cdots , \lambda_{\sigma,n} ({\bf A})).
\label{eqTheo18}
\end{equation}
Let ${\bf R}$ be the matrix:
\begin{equation}
{\bf R}= diag \left( \sqrt{\frac{\lambda_1}{\lambda_{\sigma,1} ({\bf A})}}, \cdots, \sqrt{\frac{\lambda_n}{\lambda_{\sigma,n} ({\bf A})}} , \sqrt{\frac{\lambda_1}{\lambda_{\sigma,1} ({\bf A})}}, \cdots, \sqrt{\frac{\lambda_n}{\lambda_{\sigma,n} ({\bf A})}} \right).
\label{eqTheo19}
\end{equation}
We then recover (\ref{eqTheo15}) with ${\bf P}=({\bf R} {\bf M^T})^{-1}$. If we define ${\bf \Delta}$ as in (\ref{eqTheo17}), then ${\bf P} \in \mathcal{D} (n; \delta)$ (cf. (\ref{eqTheo16})). The remaining statement is then a simple consequence of the previous theorem.
\end{proof}

Obviously, the matrix ${\bf P}$ in the previous theorem is not unique. Indeed, in the proof, we used the $\sigma$-symplectic spectrum $Spec_{\sigma} ({\bf A})$ with respect to the standard symplectic form $\sigma$. But we could equally have used the spectrum $Spec_{\omega} ({\bf A})$ with respect to any other symplectic form $\omega$. That would obviously lead to a different matrix ${\bf P}$.

We are now able to restate the $\omega$-RSUP (\ref{eqGaussian14}) in terms of the $\omega$-symplectic spectrum of the covariance matrix ${\bf A}$.

\begin{theorem}\label{TheoremOmegaRSUP1}
Let ${\bf A} \in \mathbb{P}^{\times}_{2n \times 2n} (\mathbb{R})$ and let $\lambda_{\omega,1} ({\bf A})$ denote its smallest $\omega$-Williamson invariant. Then it is the covariance matrix of a $\omega$-Wigner function if and only if
\begin{equation}
\lambda_{\omega ,1} ({\bf A}) \ge \frac{1}{2}.
\label{eqOmegaRSUP1}
\end{equation}
\end{theorem}

\begin{proof}
If ${\bf A}$ is the covariance matrix of a $\omega$-Wigner function, then it satisfies the $\omega$-RSUP (\ref{eqGaussian14}). Let ${\bf D} = diag (\lambda_{\omega ,1} ({\bf A}) , \cdots, \lambda_{\omega ,n} ({\bf A}),\lambda_{\omega , 1} ({\bf A})  , \cdots, \lambda_{ \omega , n} ({\bf A}) )$, where $\left\{\lambda_{ \omega , j} ({\bf A}) \right\}_{1 \le j \le n}$ are the $\omega$-Williamson invariants of ${\bf A}$. Thus, from (\ref{eq12.4}) we have
\begin{equation}
{\bf A} + \frac{i}{2} {\bf \Omega} \ge 0 \Leftrightarrow {\bf D}+ \frac{i}{2} {\bf S^{-1}}{\bf \Omega}{\bf S^{-1 T}} \ge 0 \Leftrightarrow {\bf D}+ \frac{i}{2} {\bf J} \ge 0.
\label{eqOmegaRSUP2}
\end{equation}
Again, the eigenvalues of ${\bf D}+ \frac{i}{2} {\bf J} $ are of the form $\lambda_{ \omega ,j} ({\bf A}) \pm \frac{1}{2}$, $j=1, \cdots, n$, and so the last inequality in (\ref{eqOmegaRSUP2}) is equivalent to (\ref{eqOmegaRSUP1}).

Conversely, if ${\bf A}$ verifies (\ref{eqOmegaRSUP1}), then it also satisfies the $\omega$-RSUP and the Gaussian measure with covariance matrix ${\bf A}$ is a Wigner function.
\end{proof}

It is worth mentioning that Narcowich-Wigner spectra of Gaussians are completely determined by the lowest Williamson invariant for every symplectic form. More specifically, we have:

\begin{theorem}\label{TheoremNWspectrumGaussians}
Let $\mathcal{G}_{{\bf A}}$ denote a Gaussian measure in phase space (\ref{eqGaussian1}) and $\omega(z,z^{\prime})= z^{\prime} \cdot {\bf \Omega^{-1}} z$ a symplectic form. Then the following statements are equivalent

\begin{enumerate}
\item $\alpha \in \mathcal{W}^{{\bf \Omega}} (\mathcal{F}\mathcal{G}_{{\bf A}})$,

\item ${\bf A} + \frac{i \alpha}{2 }{\bf \Omega} \ge 0$,

\item $\lambda_{\omega,1} ({\bf A}) \ge \frac{|\alpha|}{2}$.
\end{enumerate}

\end{theorem}

\begin{proof}
The equivalence of these statements was proven in \cite{Narcowich1} for the case $\omega=\sigma$. Let $\mathcal{G}_{{\bf A}}^{\sigma} (z)=\mathcal{G}_{{\bf A}} ({\bf S} z) = \mathcal{G}_{{\bf S^{-1}  A \left(S^{-1} \right)^T}} (z)$ for ${\bf S} \in \mathcal{D}(n; \omega)$. From Remark \ref{Remark5.1} we have:
\begin{equation}
\begin{array}{c}
\alpha \in \mathcal{W}^{{\bf \Omega}} (\mathcal{F}\mathcal{G}_{{\bf A}}) \Leftrightarrow \alpha \in \mathcal{W}^{{\bf J}} (\mathcal{F}\mathcal{G}_{{\bf S^{-1}  A \left(S^{-1} \right)^T}}) \Leftrightarrow \\
\\
{\bf S^{-1}  A \left(S^{-1} \right)^T} + \frac{i \alpha}{2} {\bf J} \ge 0 \Leftrightarrow {A } + \frac{i \alpha}{2} {\bf \Omega} \ge 0
\end{array}
\label{eqOmegaRSUP2.1}
\end{equation}
This proves the equivalence of 1 and 2. The equivalence of 2 and 3 was shown in the proof of Theorem \ref{TheoremOmegaRSUP1}.
\end{proof}

\begin{remark}\label{Remark1}
Notice that, if $\sigma \ne \pm \omega$, then in general $Spec_{\sigma} ({\bf A}) \ne Spec_{\omega} ({\bf A})$. Here is an example. For $n=2$, let
\begin{equation}
{\bf A} = \left(
\begin{array}{c c c c}
\alpha & 0 & 0 &0 \\
0 & \alpha & 0 & 0\\
0 & 0 & \beta & 0\\
0 & 0 & 0 & \beta
\end{array}
\right),
\label{eqRemark1}
\end{equation}
and
\begin{equation}
{\bf \Omega} = \left(
\begin{array}{c c c c}
0 & 0 & \theta_1 &0 \\
0 & 0 & 0 & \theta_2\\
- \theta_1 & 0 & 0 & 0\\
0 & -\theta_2 & 0 & 0
\end{array}
\right),
\label{eqRemark2}
\end{equation}
with $\alpha, \beta >0$, and $\theta_1 =\frac{1}{\theta_2} \ge 1$. After a straightforward computation, we conclude that
\begin{equation}
Spec_{\sigma} ({\bf A}) = \left( \sqrt{\alpha \beta} , \sqrt{\alpha \beta} \right),
\label{eqRemark3}
\end{equation}
while
\begin{equation}
Spec_{\omega} ({\bf A}) = \left( \frac{\sqrt{\alpha \beta}}{\theta_1} , \frac{\sqrt{\alpha \beta}}{\theta_2} \right).
\label{eqRemark4}
\end{equation}
So if $\theta_1=\theta_2=1$, then the two spectra coincide, otherwise they differ.
\end{remark}

The example in the previous remark is just a particular instance of Theorem \ref{TheoremIdenticalSpectra1} (see below). But first we recall the following Lemma which was proven in \cite{Dias5}. We denote by $Sp^+ (n; \sigma) = Sp(n; \sigma) \cap \mathbb{P}^{\times}_{2n \times 2n} (\mathbb{R})$ the set of real symmetric positive-definite symplectic $2n \times 2n$ matrices.

\begin{lemma}\label{LemmaSymplecticAntiSymplectic1}
Let ${\bf M} \in Gl(2n)$ and assume that ${\bf M}^T {\bf G} {\bf M} \in Sp(n; \sigma)$ for every
\begin{equation}
{\bf G} = \left(
\begin{array}{c c}
{\bf X} & {\bf 0}\\
{\bf 0} & {\bf X}^{-1}
\end{array}
\right) \in Sp^+ (n; \sigma).
\label{eqLemmaSymplecticAntiSymplectic1}
\end{equation}
Then ${\bf M}$ is either $\sigma$-symplectic or $\sigma$-anti-symplectic.
\end{lemma}

\begin{theorem}\label{TheoremIdenticalSpectra1}
Let $\omega_1$, $\omega_2$ denote two symplectic forms on $\mathbb{R}^{2n}$. Then $Spec_{\omega_1} ({\bf A})= Spec_{\omega_2} ({\bf A})$ for all ${\bf A} \in \mathbb{P}^{\times}_{2n \times 2n} (\mathbb{R})$ if and only if $\omega_1 = \pm \omega_2$.
\end{theorem}

\begin{proof}
Let us first assume that $\omega_1= \sigma$, $\omega_2= \omega$ and $Spec_{\sigma} ({\bf A})= Spec_{\omega} ({\bf A})$ for all ${\bf A} \in \mathbb{P}^{\times}_{2n \times 2n} (\mathbb{R})$. By the $\sigma$-Williamson Theorem (Theorem \ref{Williamsontheorem2}) and the $\omega$-Williamson Theorem (Theorem \ref{TheoremOmegaWilliamsonTheorem1}), there exist ${\bf P_{A}} \in Sp(n; \sigma)$ and ${\bf S_A} \in \mathcal{D} (n; \omega)$ such that
\begin{equation}
{\bf S_A^{-1} A S_A^{-1 T}} = {\bf P_A A P_A^T}.
\label{eqTheoremIdenticalSpectra1}
\end{equation}
Consequently
\begin{equation}
{\bf A} = {\bf S_A^{\prime} A S_A^{\prime T}},
\label{eqTheoremIdenticalSpectra2}
\end{equation}
where
\begin{equation}
{\bf S_A^{\prime}} = {\bf S_A P_A} \in \mathcal{D} (n; \omega).
\label{eqTheoremIdenticalSpectra3}
\end{equation}
 From (\ref{eqTheoremIdenticalSpectra2}) we conclude that there exists ${\bf S_A^{\prime}} \in \mathcal{D} (n; \omega)$ such that:
\begin{equation}
{\bf S_A^{\prime}} = {\bf A} ({\bf S_A^{\prime T}})^{-1} {\bf A}^{-1}.
\label{eqTheoremIdenticalSpectra4}
\end{equation}
Since this holds for all ${\bf A} \in \mathbb{P}_{2n \times 2n}^{\times} (\mathbb{R})$, in particular it is true for ${\bf A} \in Sp^+ (n; \sigma)$. Consequently:
\begin{equation}
\begin{array}{c}
{\bf \Omega} = {\bf S_A^{\prime} J} {\bf S_A^{\prime}}^T = {\bf A} ({\bf S_A^{\prime T}})^{-1} {\bf A}^{-1} {\bf J} ({\bf A}^{-1})^T ({\bf S_A^{\prime }})^{-1} {\bf A} =\\
\\
={\bf A} ({\bf S_A^{\prime T}})^{-1} {\bf J} ({\bf S_A^{\prime }})^{-1} {\bf A}
= - {\bf A} \left( {\bf S_A^{\prime} J} {\bf S_A^{\prime}}^T \right)^{-1} {\bf A } =\\
\\
= -{\bf A}{\bf \Omega}^{-1}{\bf A}
\end{array}
\label{eqTheoremIdenticalSpectra5}
\end{equation}
So, we have proven that
\begin{equation}
{\bf \Omega}^T= {\bf A}{\bf \Omega}^{-1}{\bf A}
\label{eqTheoremIdenticalSpectra6}
\end{equation}
for all ${\bf A} \in Sp^+ (n; \sigma)$.

Next, fix an arbitrary ${\bf S} \in \mathcal{D} (n; \omega)$. From (\ref{eqTheoremIdenticalSpectra5}) and (\ref{eqsymplecticform9}) we have:
\begin{equation}
{\bf S} {\bf J} {\bf S}^T = {\bf \Omega}= - {\bf A} ({\bf S} {\bf J} {\bf S}^T )^{-1} {\bf A} \Leftrightarrow {\bf S}^{-1}  {\bf A} ({\bf S}^T)^{-1} {\bf J} \left( {\bf S}^{-1}  {\bf A} ({\bf S}^T)^{-1}\right)^T =  {\bf J},
\label{eqTheoremIdenticalSpectra7}
\end{equation}
for all ${\bf A} \in Sp^+ (n; \sigma)$. In other words ${\bf S}^{-1}  {\bf A} ({\bf S}^T)^{-1} \in Sp^+ (n; \sigma)$ for all ${\bf A} \in Sp^+ (n; \sigma)$. But from Lemma \ref{LemmaSymplecticAntiSymplectic1}, this is possible if and only if ${\bf S}^{-1}$ is $\sigma$-symplectic or $\sigma$-antisymplectic. In the first case, we have $\omega= \sigma$ and in the second case $\omega= - \sigma$.

Next consider arbitrary symplectic forms $\omega_1 (z,z^{\prime})=z \cdot {\bf \Omega_1}^{-1} z^{\prime}$, $\omega_2 (z,z^{\prime})=z \cdot {\bf \Omega_2}^{-1} z^{\prime}$ such that $Spec_{\omega_1} ({\bf A})= Spec_{\omega_2} ({\bf A})$ for all ${\bf A} \in \mathbb{P}^{\times}_{2n \times 2n} (\mathbb{R})$. Now let ${\bf S_1} \in \mathcal{D} (n; \omega_1)$. For any  ${\bf A} \in \mathbb{P}^{\times}_{2n \times 2n} (\mathbb{R})$, we have that  ${\bf S_1}{\bf A} {\bf S_1}^T\in \mathbb{P}^{\times}_{2n \times 2n} (\mathbb{R})$. On the other hand, it is easy to verify that (cf. (\ref{eq12.7}))
\begin{equation}
 Spec_{\omega_1} ({\bf S_1}{\bf A} {\bf S_1}^T)=  Spec_{\sigma} ({\bf A}),
\label{eqTheoremIdenticalSpectra8}
\end{equation}
and
\begin{equation}
 Spec_{\omega_2} ({\bf S_1}{\bf A} {\bf S_1}^T)=  Spec_{\omega} ({\bf A}),
\label{eqTheoremIdenticalSpectra9}
\end{equation}
with
\begin{equation}
\omega (z,z^{\prime})=z \cdot {\bf \Omega}^{-1} z^{\prime} \hspace{1 cm} {\bf \Omega} = {\bf S_1}^{-1} {\bf \Omega_2} ({\bf S_1}^{-1})^T.
\label{eqTheoremIdenticalSpectra10}
\end{equation}
We conclude that
\begin{equation}
 Spec_{\omega_1} ({\bf S_1}{\bf A} {\bf S_1}^T)=  Spec_{\omega_2} ({\bf S_1}{\bf A} {\bf S_1}^T) \Leftrightarrow Spec_{\sigma} ({\bf A}) =  Spec_{\omega} ({\bf A}),
\label{eqTheoremIdenticalSpectra11}
\end{equation}
for all ${\bf A} \in \mathbb{P}^{\times}_{2n \times 2n} (\mathbb{R})$. From the previous analysis $\omega = \pm \sigma$ or, equivalently, $\omega_1 = \pm \omega_2$.
\end{proof}

Let us now prove the following refinement of Theorem \ref{TheoremIdenticalSpectra1}. We have proven that for $\omega_2 \ne \pm \omega_1$, we can always find matrices ${\bf A} \in \mathbb{P}^{\times}_{2n \times 2n} (\mathbb{R})$ with $Spec_{\omega_1} ({\bf A}) \ne Spec_{\omega_2} ({\bf A})$. We now show that in fact the smallest $\omega_1$- and $\omega_2$-Williamson invariants cannot coincide for all ${\bf A} \in \mathbb{P}^{\times}_{2n \times 2n} (\mathbb{R})$.

\begin{theorem}\label{TheoremUnboudedQuotient1}
Let $\omega_1, \omega_2$ be symplectic forms on $\mathbb{R}^{2n}$ such that $\omega_1 \ne \pm \omega_2$. For any $K>0$ there exists a matrix ${\bf A} \in \mathbb{P}^{\times}_{2n \times 2n} (\mathbb{R})$ such that
\begin{equation}
\frac{\lambda_{\omega_1,1} ({\bf  A})}{\lambda_{\omega_2,1} ({\bf A})} \ge K.
\label{eqUnboudedQuotient1}
\end{equation}
\end{theorem}

\begin{proof}
Following an argument similar to that in Theorem \ref{TheoremIdenticalSpectra1}, we may assume without loss of generality that $\omega_1 = \sigma$ and set $\omega_2 = \omega$. Suppose that there exists $K>0$ such that
\begin{equation}
\frac{\lambda_{\sigma,1} ( {\bf A})}{\lambda_{\omega,1} ({\bf A})} < K,
\label{eqUnboudedQuotient2}
\end{equation}
for all ${\bf A} \in \mathbb{P}^{\times}_{2n \times 2n} (\mathbb{R})$.

It follows from items 2 and 3 of Theorem \ref{TheoremNWspectrumGaussians} (setting $\alpha = \pm \lambda_{\sigma,1} ({\bf A})$ and ${\bf \Omega}={\bf J}$) that
\begin{equation}
{\bf A} \pm i \lambda_{\sigma,1} ({\bf A}) {\bf J} \ge 0,
\label{eqUnboudedQuotient3}
\end{equation}
while (setting $\alpha = \pm \frac{\lambda_{\sigma,1} ({\bf A})}{K}$, and taking (\ref{eqUnboudedQuotient2}) into account)
\begin{equation}
{\bf A} \pm i \frac{\lambda_{\sigma,1} ({\bf A})}{K} {\bf \Omega} \ge 0.
\label{eqUnboudedQuotient4}
\end{equation}
From the previous inequality, we also have:
\begin{equation}
{\bf A} +i \lambda_{\sigma,1} ({\bf A}) {\bf J} +  i\lambda_{\sigma,1} ({\bf A}) \left(\pm \frac{ {\bf \Omega}}{K} - {\bf J} \right) \ge 0.
\label{eqUnboudedQuotient5}
\end{equation}
Let $u_1=e_1 + i f_1$ be an eigenvector of ${\bf A J^{-1}} $ associated with the eigenvalue $-i \lambda_{\sigma,1} ({\bf A})$, i.e.
\begin{equation}
{\bf AJ^{-1}} u_1 = -i \lambda_{\sigma,1} ({\bf A}) u_1.
\label{eqUnboudedQuotient6}
\end{equation}
Multiplying (\ref{eqUnboudedQuotient5}) on the left by $\overline{ u_1}\cdot {\bf J}$, on the right by ${\bf J^T} u_1={\bf J^{-1}} u_1$, and taking into account that $\lambda_{\sigma,1} ({\bf A}) >0$ yields:
\begin{equation}
i \overline{u_1} \cdot {\bf J} \left(\pm \frac{ {\bf \Omega}}{K} - {\bf J} \right) {\bf J^{-1}} u_1 \ge 0.
\label{eqUnboudedQuotient7}
\end{equation}
If we define the symplectic form
\begin{equation}
\eta (z, z^{\prime}):= \frac{1}{K} z \cdot {\bf J \Omega J^T} z^{\prime} ,
\label{eqUnboudedQuotient8}
\end{equation}
for $z,z^{\prime} \in \mathbb{R}^{2n}$, we get from (\ref{eqUnboudedQuotient7}):
\begin{equation}
\pm \eta (f_1,e_1) + \sigma (f_1,e_1) \ge 0.
\label{eqUnboudedQuotient9}
\end{equation}
So the previous equation holds for all vectors $e_1,f_1$ such that $u_1= e_1 + i f_1$ is an eigenvector of $i {\bf A J^{-1}}$ for some matrix ${\bf A} \in  \mathbb{P}^{\times}_{2n \times 2n} (\mathbb{R})$ with eigenvalue equal to the smallest Williamson invariant of ${\bf A}$. By Proposition \ref{PropostionEigenvector} such a matrix exists for all vectors $e, f \in \mathbb{R}^{2n}$ such that $\sigma (f,e) >0$. So, in fact, we have
\begin{equation}
|\eta (f,e)| \le  \sigma (f,e) ,
\label{eqUnboudedQuotient10}
\end{equation}
for all vectors $e,f$ such that $\sigma (f,e) >0$. If $\sigma (f,e) <0$, then (since $\sigma (f,e) = -\sigma (e,f) )$, it follows that
\begin{equation}
|\eta (f,e)| \le  \sigma (e,f) .
\label{eqUnboudedQuotient11}
\end{equation}
Altogether, we may write
\begin{equation}
|\eta (f,e)| \le  |\sigma (f,e)| ,
\label{eqUnboudedQuotient12}
\end{equation}
whenever $\sigma (f,e) \ne 0$. Suppose that $\sigma (f,e) = 0$. We may regard $f$ and $e$ as elements of some Lagrangian plane \cite{Cannas}. Since Lagrangian planes have no interior points, we can find a sequence $\left\{e_n \right\}_{n \in \mathbb{N}}$ converging to $e$, such that $\sigma (f,e_n) \ne 0$. But from (\ref{eqUnboudedQuotient12}), we have:
\begin{equation}
|\eta (f,e_n)| \le  |\sigma (f,e_n)|.
\label{eqUnboudedQuotient12.1}
\end{equation}
However, since symplectic forms are continuous, inequality (\ref{eqUnboudedQuotient12}) must hold for all $e,f \in \mathbb{R}^{2n}$. By Proposition \ref{PropositionSymplecticInequality}, there exists a constant $0< |a | \le 1$ such that
\begin{equation}
\eta = a \sigma .
\label{eqUnboudedQuotient13}
\end{equation}
From (\ref{eqUnboudedQuotient8}), we have:
\begin{equation}
{\bf \Omega} = - a  K {\bf J}.
\label{eqUnboudedQuotient14}
\end{equation}
Since the symplectic form $\det ({\bf \Omega})=1$, we conclude that $\omega = \pm \sigma$ and we have a contradiction. Thus (\ref{eqUnboudedQuotient2}) cannot hold.
\end{proof}

\subsection{The landscape of Wigner functions}

As an application of the previous result, we prove the following Theorem on the symplectic covariance of Wigner functions.

\begin{theorem}\label{TheoremSymplecticCovariance}
Let ${\bf M} \in Gl(2n; \mathbb{R})$. Then the operator
\begin{equation}
F (z) \mapsto (\mathcal{U}_{{\bf M}} F)(z) := | \det{\bf M}| F (Mz)
\label{eqSymplecticCovariance1}
\end{equation}
maps (pure or mixed) $\omega$-Wigner functions to $\omega$-Wigner functions if and only if ${\bf M}$ is either $\omega$-symplectic or $\omega$-anti-symplectic.
\end{theorem}

\begin{proof}
One calls a transformation which maps quantum states to quantum states a quantum channel or a quantum dynamical map.

We start by considering the case $\omega =\pm \sigma$. It is a well documented fact that $\mathcal{U}_{{\bf M}}$ is a quantum dynamical map if ${\bf M}$ is a $\sigma$-symplectic or $\sigma$-anti-symplectic matrix \cite{Dias5}.

Conversely, suppose that $\mathcal{U}_{{\bf M}}$ is a quantum dynamical map. Then, in particular it maps every Gaussian $\sigma$-Wigner function $\mathcal{G}_{{\bf A}}$ with covariance matrix ${\bf A}$ to another Gaussian $\sigma$-Wigner function with covariance matrix ${\bf M^{-1}} {\bf A} ({\bf M^{-1}})^T$.

For every ${\bf A} \in \mathbb{P}_{2n \times 2n}^{\times} (\mathbb{R})$, the matrix $\frac{1}{2 \lambda_{\sigma,1} ({\bf A})} {\bf A} $ satisfies the $\sigma$-RSUP. Indeed, acccording to Lemma \ref{LemmaRescalings1}, its smallest Williamson invariant is $\frac{1}{2 \lambda_{\sigma,1} ({\bf A})} \lambda_{\sigma,1} ({\bf A})=\frac{1}{2}$. Hence the matrix $\frac{1}{2 \lambda_{\sigma,1} ({\bf A})} {\bf M^{-1}} {\bf A} ({\bf M^{-1}})^T$ must also satisfy the $\sigma$-RSUP:
\begin{equation}
\frac{1}{2 \lambda_{\sigma,1} ({\bf A})} {\bf M^{-1}} {\bf A} ({\bf M^{-1}})^T + \frac{i}{2} {\bf J} \ge 0 \Leftrightarrow {\bf A} + i \lambda_{\sigma,1} ({\bf A}){\bf M} {\bf J} {\bf M}^T \ge 0.
\label{eqWignerCovariance2}
\end{equation}
Next define ${\bf \Omega} = \gamma^{-1} {\bf M} {\bf J} {\bf M}^T$, where $\gamma = | \det {\bf M}|^{1/n}$. Thus $\omega (z,z^{\prime}) = z^T {\bf \Omega^{-1}} z^{\prime}$ is a normalized symplectic form. From (\ref{eqWignerCovariance2}), we have:
\begin{equation}
{\bf A} + i \mu {\bf \Omega} \ge 0,
\label{eqWignerCovariance3}
\end{equation}
where $\mu =  \lambda_{\sigma,1} ({\bf A}) \gamma$. Consequently:
\begin{equation}
\lambda_{\omega,1} ({\bf A}) \ge \mu \Leftrightarrow \frac{\lambda_{\sigma,1} ({\bf A})}{\lambda_{\omega,1} ({\bf A})} \le \frac{1}{\gamma},
\label{eqWignerCovariance4}
\end{equation}
for all ${\bf A} \in \mathbb{P}_{2n \times 2n}^{\times} (\mathbb{R})$. But according to Theorem \ref{TheoremUnboudedQuotient1}, this can happen only if $\omega = \pm \sigma$, or ${\bf \Omega} = \pm {\bf J}$. This means that $\gamma^{-1/2} {\bf M}$ is a $\sigma$-symplectic or a $\sigma$-anti-symplectic matrix. Consequently, if $\mathcal{W}^{{\bf J}} (\mathcal{F}W \rho)$ denotes the Narcowich-Wigner spectrum of some Wigner function $W \rho$, then the Narcowich-Wigner spectrum of $\mathcal{U}_M (W \rho)$ becomes $\mathcal{W}^{{\bf J}} \left( \mathcal{F}\mathcal{U}_M (W \rho)\right)= \gamma^{-1} \mathcal{W}^{{\bf J}}  (\mathcal{F}W \rho)$. Let $\psi \in L^2 (\mathbb{R}^n)$ be some non-Gaussian state. Then $\mathcal{W}^{{\bf J}} (\mathcal{F}W \psi) = \left\{ - 1, 1 \right\}$ (see Theorem \ref{TheoremPropertiesNWSpectra}, item 6). It follows that $\mathcal{W}^{{\bf J}}  (\mathcal{F} \mathcal{U}_M W \psi) = \left\{ - \frac{1}{\gamma}, \frac{1}{\gamma} \right\}$. And thus, $\mathcal{U}_M W \psi$ is a Wigner function if and only if $\gamma =1$.

Next, assume that $\omega \ne \pm \sigma$. Let $W^{\omega} \rho$ be an arbitrary $\omega$-Wigner function. Then, given a Darboux matrix ${\bf S} \in \mathcal{D}(n; \omega)$, the function
\begin{equation}
W^{\sigma} \rho (z)= W^{\omega} \rho ({\bf S} z)
\label{eqWignerCovariance4.1}
\end{equation}
is a $\sigma$-Wigner function.

Define the matrix ${\bf P}$ by way of
\begin{equation}
{\bf P}:= {\bf S^{-1} M S}.
\label{eqWignerCovariance4.2}
\end{equation}
We then have that $\left(\mathcal{U}_{{\bf M}} W^{\omega} \rho\right) (z)$ is a $\omega$-Wigner function if and only if $\left(\mathcal{U}_{{\bf M}} W^{\omega} \rho \right) ({\bf S} z)$ is a $\sigma$-Wigner function. On the other hand:
\begin{equation}
\begin{array}{c}
\left(\mathcal{U}_{{\bf M}} W^{\omega} \rho \right) ({\bf S} z)= | \det {\bf M} | W^{\omega} \rho ({\bf MS} z) = | \det {\bf M} | W^{\sigma} \rho ({\bf S^{-1} MS} z) =\\
 \\
 =| \det {\bf P} | W^{\sigma} \rho ({\bf P} z) = \left(\mathcal{U}_{{\bf P}} W^{\sigma} \rho \right) ( z)
\end{array}
\label{eqWignerCovariance4.3}
\end{equation}
Thus, $\mathcal{U}_{{\bf M}} W^{\omega} \rho$ is a $\omega$-Wigner function for every density matrix $\widehat{\rho}$ if and only if $\mathcal{U}_{{\bf P}} W^{\sigma} \rho$  is a $\sigma$-Wigner function for every density matrix $\widehat{\rho}$. This, in turn, happens if and only if ${\bf P}$ is $\sigma$-symplectic or $\sigma$-anti-symplectic. But, from (\ref{eqWignerCovariance4.2}), this is equivalent to ${\bf M} $ being $\omega$-symplectic or $\omega$-anti-symplectic.
\end{proof}

\begin{remark}\label{Remarksymplecticcovariance}
This theorem is a small improvement on the results of \cite{Dias5} for $\omega= \pm \sigma$. Indeed, here we admit {\it a priori} an arbitrary value of $\det {\bf M} (\ne 0)$. In \cite{Dias5} we assumed $\det {\bf M}=1$ {\it ab initio}.
\end{remark}

We now provide a characterization of the sets of classical and quantum states with arbitrary symplectic structure.

Let $\mathcal{L} (\mathbb{R}^{2n})$ denote the set of Liouville measures on $\mathbb{R}^{2n}$. These are the continuous, pointwise non-negative functions normalized to unity. In other words, they are classical probability densities on the phase-space. Moreover, let $\mathcal{F}^{\omega_1} (\mathbb{R}^{2n})$ and $\mathcal{F}^{\omega_2} (\mathbb{R}^{2n})$ be the convex sets of $\omega_1$- and $\omega_2$-Wigner functions for some symplectic forms $\omega_1$, $\omega_2$ such that $\omega_2 \ne \pm \omega_1$. If the dimension is clear, we shall simply write $\mathcal{L},\mathcal{F}^{\omega_1},\mathcal{F}^{\omega_2}$.

In \cite{Bastos2} we proved that the sets
\begin{equation}
\begin{array}{l}
\mathcal{A}_1 = \mathcal{F}^{\omega_1}  \backslash \left( \mathcal{F}^{\omega_2} \cup \mathcal{L} \right)\\
\mathcal{A}_2 = \mathcal{F}^{\omega_2}  \backslash \left( \mathcal{F}^{\omega_1} \cup \mathcal{L} \right)\\
\mathcal{A}_3 = \mathcal{L}  \backslash \left( \mathcal{F}^{\omega_1} \cup \mathcal{F}^{\omega_2} \right)\\
\mathcal{A}_4 = \left( \mathcal{F}^{\omega_1} \cap \mathcal{F}^{\omega_2} \right) \backslash  \mathcal{L} \\
\mathcal{A}_5 = \left( \mathcal{F}^{\omega_1} \cap \mathcal{L} \right) \backslash  \mathcal{F}^{\omega_2} \\
\mathcal{A}_6 = \left( \mathcal{F}^{\omega_2} \cap \mathcal{L} \right) \backslash  \mathcal{F}^{\omega_1}\\
\mathcal{A}_7 = \mathcal{F}^{\omega_1} \cap \mathcal{F}^{\omega_2} \cap  \mathcal{L}
\end{array}
\label{eqSets1}
\end{equation}
are all non-empty, when $n=2$, $\omega_1 = \sigma$ and $\omega_2 = \omega$ given by (\ref{eqNCQM1}).

Here, we generalize the result for arbitrary dimension and symplectic forms $\omega_1 \ne \pm \omega_2$.
%%%%%%%%%.%%%%%%%%%%%%%%%%%%%%%%%%%%%%%

\begin{figure}
\begin{center}
\vspace{-4cm} \hspace{-4cm}\includegraphics [scale=0.5]{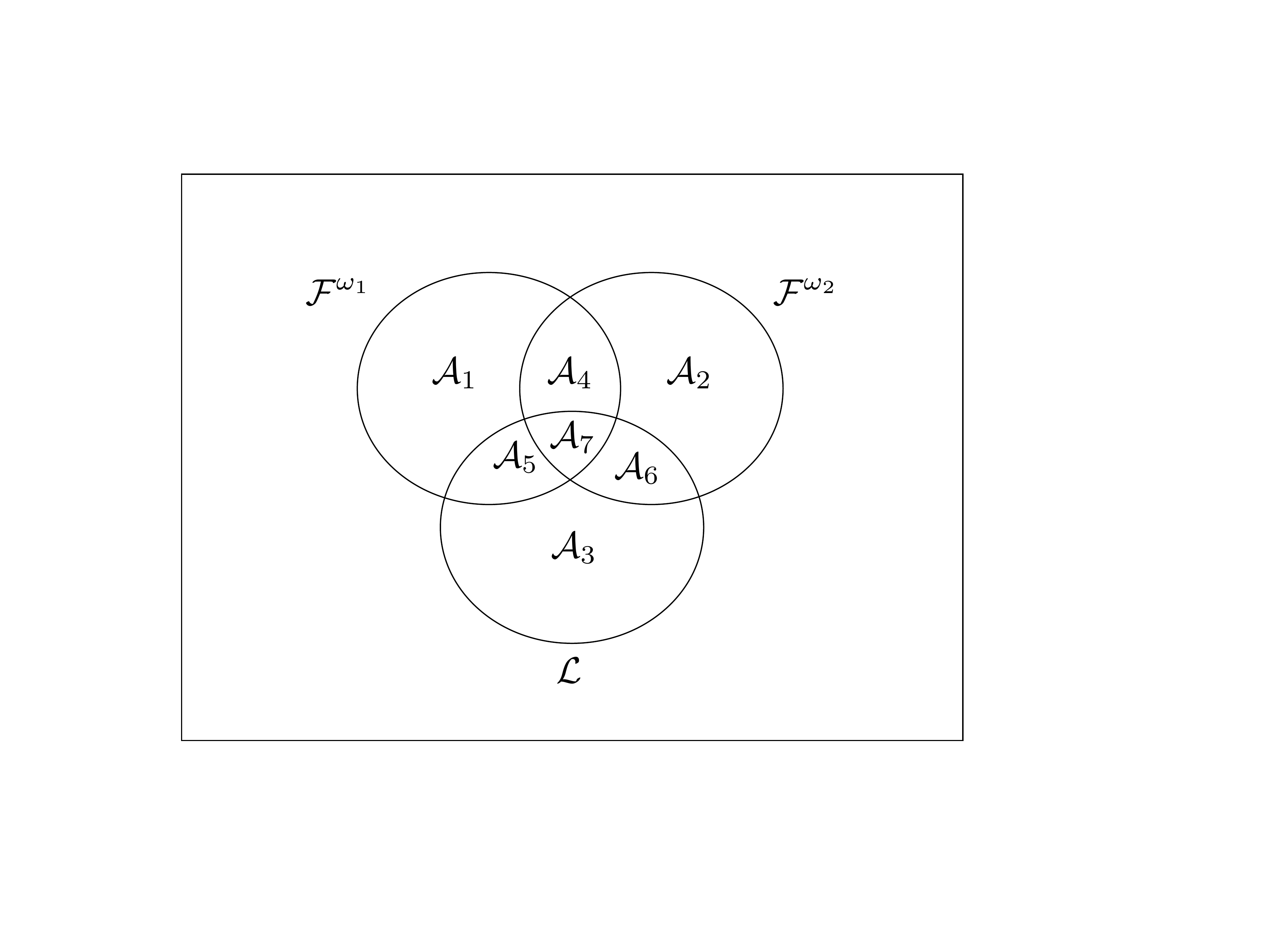}
\caption{Different sets of functions and their intersections.}
\end{center}
\label{diagrama}
\end{figure}
%%%%%%%%%%%%%%%%%%%%%%%%%%%%%%%%%%%%%%

\begin{theorem}\label{TheoremSets1}
Let $\omega_1$, $\omega_2$ be arbitrary symplectic forms on
$\mathbb{R}^{2n}$ with $\omega_1 \ne \pm \omega_2$. Then the sets
$\mathcal{A}_i$, $i=1, \cdots,7$ defined as in (\ref{eqSets1}) are
all non-empty (see Figure $1$).
\end{theorem}

\begin{proof}
As in \cite{Bastos2} we shall suggest ways to construct examples of functions in each of these sets. We start with

\noindent
${\bf \mathcal{A}_3:}$ Consider an arbitrary Gaussian function $\mathcal{G}_{{\bf A}}$ with covariance matrix ${\bf A}$. Clearly $\mathcal{G}_{{\bf A}} \in \mathcal{L} $. Compute the $\omega_1$- and $\omega_2$-symplectic spectra of ${\bf A}$. Let $\lambda_{max,1} ({\bf A})= max \left\{ \lambda_{\omega_1,1} ({\bf A}), \lambda_{\omega_2,1} ({\bf A})\right\}$. If $\lambda_{max,1} ({\bf A})< \frac{1}{2}$, then $\mathcal{G}_{{\bf A}} \notin \mathcal{F}^{\omega_1} \cup \mathcal{F}^{\omega_2}$ and we are done. If $\lambda_{max,1} ({\bf A}) \ge \frac{1}{2}$, then choose some $0 < \mu < \frac{1}{2 \lambda_{max,1}({\bf A})}$. According to Lemma \ref{LemmaRescalings1}, the Gaussian $\mathcal{G}_{\mu {\bf A}}$ is such that $\lambda_{max,1} (\mu {\bf A})< \frac{1}{2}$ and thus $\mathcal{G}_{\mu {\bf A}} \notin \mathcal{F}^{\omega_1} \cup \mathcal{F}^{\omega_2}$.

\vspace{0.3 cm}\noindent
${\bf \mathcal{A}_7:}$ Again consider a Gaussian $\mathcal{G}_{{\bf A}}$ with covariance matrix ${\bf A}$ and let $\lambda_{min,1} ({\bf A})= min \left\{ \lambda_{\omega_1,1} ({\bf A}), \lambda_{\omega_2,1} ({\bf A})\right\}$. If $\lambda_{min,1} ({\bf A}) \ge \frac{1}{2}$, then $\mathcal{G}_{{\bf A}} \in \mathcal{A}_7$. Otherwise, let $\mu \ge \frac{1}{2 \lambda_{min,1} ({\bf A})}$. Then $\lambda_{min,1} (\mu {\bf A})\ge  \frac{1}{2}$ and $\mathcal{G}_{\mu {\bf A}} \in \mathcal{A}_7$.

\vspace{0.3 cm}\noindent
${\bf \mathcal{A}_5:}$ According to Theorem \ref{TheoremUnboudedQuotient1}, there exists a Gaussian $\mathcal{G}_{{\bf A}}$ with covariance matrix ${\bf A}$, such that $\lambda_{\omega_1,1} ({\bf A}) > \lambda_{\omega_2,1} ({\bf A})$. Let $\mu$ be such that $2\lambda_{\omega_1,1} ({\bf A}) > \mu^{-1} > 2 \lambda_{\omega_2,1} ({\bf A})$. Then the Gaussian $\mathcal{G}_{\mu {\bf A}}$ is such that $\lambda_{\omega_1,1} (\mu {\bf A}) \ge \frac{1}{2}$, while $\lambda_{\omega_2,1} (\mu {\bf A}) < \frac{1}{2}$. Hence $\mathcal{G}_{\mu {\bf A}} \in \mathcal{A}_5$.

\vspace{0.3 cm}\noindent
${\bf \mathcal{A}_6:}$ The result follows the same steps as in $\mathcal{A}_5$, with the replacement $\omega_1 \leftrightarrow \omega_2$.

\vspace{0.3 cm}\noindent
${\bf \mathcal{A}_1:}$ Let $h \in \mathcal{F}^{\sigma} \backslash \mathcal{L}$, that is a $\sigma$-Wigner function which is not everywhere nonnegative. Then the function $h_1(z) = h({\bf S_1}^{-1}z)$ with ${\bf S_1} \in \mathcal{D}(n; \omega_1)$ is in $\mathcal{F}^{\omega_1} \backslash \mathcal{L}$. Choose $h_5 \in \mathcal{A}_5$. Since $h_5 \notin \mathcal{F}^{\omega_2}$, there exists $g_2 \in \mathcal{F}^{\omega_2} $, such that (cf.(\ref{eqQuantumConditions2})):
\begin{equation}
b := \int_{\mathbb{R}^{2n}} h_5(z) g_2(z)dz < 0.
\label{eqSets2}
\end{equation}
Let also
\begin{equation}
a := \int_{\mathbb{R}^{2n}} h_1(z) g_2(z)dz .
\label{eqSets3}
\end{equation}
First suppose that $a \le 0$. Let $z_0 \in \mathbb{R}^{2n}$ be such that $h_1 (z_0) < 0 < h_5 (z_0)$. Such a $z_0$ can always be found because $h_5 (z) \ge 0$ for all $z$ and is not identically zero. On the other hand if $h_1 (z) \in \mathcal{F}^{\omega_1} \backslash \mathcal{L}$, then  $h_1 (z - \zeta) \in \mathcal{F}^{\omega_1} \backslash \mathcal{L}$ for any fixed $\zeta$. So we can always translate $h_1$ so that $h_1 (z_0) < 0$. Next, choose $1>p>0$ such that
\begin{equation}
\frac{p}{1-p} > \frac{h_5(z_0)}{|h_1(z_0)|}.
\label{eqSets4}
\end{equation}
The function
\begin{equation}
f_1(z) = p h_1(z) + (1-p) h_5(z)
\label{eqSets5}
\end{equation}
being a convex combination of $h_1, h_5 \in \mathcal{F}^{\omega_1}$ also belongs to $\mathcal{F}^{\omega_1}$. Moreover, since by (\ref{eqSets4}), $f_1(z_0)<0$, we also have $f_1 \notin \mathcal{L}$. Finally from (\ref{eqSets2},\ref{eqSets3}):
\begin{equation}
\int_{\mathbb{R}^{2n}} f_1(z) g_2(z)dz =p a + (1-p)b <0,
\label{eqSets6}
\end{equation}
and thus $f_1 \notin  \mathcal{F}^{\omega_2}$. Altogether, $f_1 \in \mathcal{A}_1$.

Now suppose that we have instead $a > 0$. We start by showing that, by translating $h_1(z) \mapsto T_{\zeta} h_1(z) = h_1(z- \zeta)$ appropriately , we can always find $z_1 \in \mathbb{R}^{2n}$ such that
\begin{equation}
T_{\zeta} h_1 (z_1)  < 0< \frac{h_5(z_1)}{|b|} < \frac{|\left(T_{\zeta}h_1 \right)(z_1)|}{a}.
\label{eqSets7}
\end{equation}
Indeed, let $z_2 \in \mathbb{R}^{2n}$ be such that $h_1(z_2) <0$ and $z_1$ such that
\begin{equation}
0< h_5(z_1) < (2 \pi)^n |b h_1(z_2)|.
\label{eqSets8}
\end{equation}
It is always possible to find such a $z_1$ because we may assume, without compromising our argument, that $h_5 \in \mathcal{S} ( \mathbb{R}^{2n})$. If $\zeta = z_1-z_2$, then we have from equation (\ref{eqSets8}), the Cauchy-Schwartz inequality, the purity condition (\ref{eqMoyalidentity4}), and by replacing $h_1$ by $T_{\zeta} h_1$ in (\ref{eqSets3}):
\begin{equation}
\begin{array}{c}
0 < \frac{h_5(z_1)}{|b| } a = \frac{h_5(z_1)}{|b| } \left| \int_{\mathbb{R}^{2n}} T_{\zeta} h_1 (z) g_2 (z) \right| \le\\
\\
\le \frac{h_5(z_1)}{|b| }  ||T_{\zeta} h_1 ||_{L^2}  ||g_2||_{L^2} \le \frac{h_5(z_1)}{(2 \pi)^n |b| } < |h_1(z_2)| = | T_{\zeta} h_1 (z_1)|,
\end{array}
\label{eqSets9}
\end{equation}
which proves (\ref{eqSets7}) for $ T_{\zeta} h_1 $.

From (\ref{eqSets7}) we can choose $p \in \mathbb{R}$ such that
\begin{equation}
0 < \frac{\alpha}{1+ \alpha} <p< \frac{\beta}{1+ \beta} <1,
\label{eqSets10}
\end{equation}
with
\begin{equation}
0< \alpha =\frac{h_5(z_1)}{|\left(T_{\zeta} h_1 \right)(z_1)|} < \frac{|b|}{a}= \beta.
 \label{eqSets10.1}
\end{equation}
This is because the function $x \in \mathbb{R}^+ \mapsto \frac{x}{1+x}$ is strictly increasing.

Next, let $f_1 \in \mathcal{F}^{\omega_1}$ be defined as in (\ref{eqSets5}) with $h_1$ replaced by $T_{\zeta} h_1$. Since
\begin{equation}
f_1(z_1) = -p |\left(T_{\zeta} h_1 \right)(z_1)| + (1-p) h_5(z_1) <0
\label{eqSets11}
\end{equation}
we have $f_1 \notin \mathcal{L}$. Similarly:
\begin{equation}
\int_{\mathbb{R}^{2n}} f_1(z) g_2(z) dz = p a - (1-p) |b| <0,
\label{eqSets12}
\end{equation}
and thus $f_1 \notin \mathcal{F}^{\omega_2}$. Altogether $f_1 \in \mathcal{A}_1$.

\vspace{0.3 cm}\noindent
${\bf \mathcal{A}_2:}$ The result follows the same steps as in $\mathcal{A}_1$, with the replacement $\omega_1 \leftrightarrow \omega_2$.

\vspace{0.3 cm}\noindent
${\bf \mathcal{A}_4:}$ To simplify the argument, we start by noticing that proving that $f_4 \in \mathcal{A}_4$ is equivalent to proving that $g_4 \in \left( \mathcal{F}^{\sigma} \cap  \mathcal{F}^{\omega} \right) \backslash  \mathcal{L}$ where $g_4(z) = f_4 ({\bf S_1}z)$, $\omega (z,z^{\prime}) = z \cdot {\bf \Omega}^{-1} z^{\prime}$, ${\bf \Omega} = {\bf S_1}^{-1} {\bf \Omega_2} ({\bf S_1}^{-1})^T$ with ${\bf S_1} \in \mathcal{D} (n; \omega_1)$, ${\bf S_2} \in \mathcal{D} (n; \omega_2)$. Indeed, $f_4 \notin \mathcal{L} \Leftrightarrow g_4 \notin \mathcal{L}$. Moreover, $f_4 (z) \in \mathcal{F}^{\omega_1} \Leftrightarrow g(z) =f_4 ({\bf S_1} z) \in \mathcal{F}^{\sigma}$. Finally, $f_4 (z) \in \mathcal{F}^{\omega_2} \Leftrightarrow f_4 ({\bf S_2} z) = g_4 ({\bf S_1}^{-1} {\bf S_2} z) \in \mathcal{F}^{\sigma}$. But since ${\bf S_1}^{-1} {\bf S_2} \in \mathcal{D} (n; \omega)$, this is equivalent to $g_4 \in \mathcal{F}^{\omega}$.

Let us now devise a way to construct a function $g_4 \in \left( \mathcal{F}^{\sigma} \cap  \mathcal{F}^{\omega} \right) \backslash  \mathcal{L}$. Consider the function
\begin{equation}
\gamma (\alpha)= \det ({\bf J} - \alpha {\bf \Omega} ).
\label{eqSets13}
\end{equation}
We claim that there exists $ \alpha_0 >1$ such that
\begin{equation}
\gamma_0 := \gamma (\alpha_0) >0.
\label{eqSets14}
\end{equation}
Indeed, $\gamma (\alpha)=\det \left({\bf J \Omega^{-1}} - \alpha {\bf I} \right)$. Thus $\gamma$ is the characteristic polynomial of the matrix ${\bf J \Omega^{-1}}$. We thus have:
\begin{equation}
\begin{array}{c}
\gamma (\alpha)=\alpha^{2n} + Tr \left({\bf J \Omega^{-1}} \right) \alpha^{2n-1} + \cdots + \det \left({\bf J \Omega^{-1}} \right) = \\
\\
=\alpha^{2n} + Tr \left({\bf J \Omega^{-1}} \right) \alpha^{2n-1} + \cdots +1.
\end{array}
\label{eqSets14.1}
\end{equation}
Hence $\gamma (\alpha) \simeq \alpha^{2n}$ as $\alpha \to + \infty$, which proves the claim.
Define ${\bf \Sigma_1} = {\bf \Omega}$ and ${\bf \Sigma_2} = \gamma_0^{-\frac{1}{2n}} ({\bf J} - \alpha_0 {\bf \Omega})$ and the corresponding normalized symplectic forms:
\begin{equation}
\eta_1 = \omega, \hspace{1 cm} \eta_2 (z, z^{\prime}) = z \cdot {\bf \Sigma_2^{-1}} z^{\prime}.
\label{eqSets14.1}
\end{equation}
We claim that there exists ${\bf A} \in \mathbb{P}^{\times}_{2n \times 2n} (\mathbb{R})$ such that
\begin{equation}
\lambda_{\eta_1,1} ({\bf A}) = \alpha_0 -1,
\label{eqSets15}
\end{equation}
and
\begin{equation}
\lambda_{\eta_2,1} ({\bf A}) \ge \gamma_0^{\frac{1}{2n}}.
\label{eqSets15.1}
\end{equation}
Indeed, from Theorem \ref{TheoremUnboudedQuotient1}, we can find ${\bf A} \in \mathbb{P}^{\times}_{2n \times 2n} (\mathbb{R})$ such that
\begin{equation}
\frac{\lambda_{\eta_2,1} ({\bf A})}{\lambda_{\eta_1,1} ({\bf A})} \ge \frac{\gamma_0^{\frac{1}{2n}}}{\alpha_0 -1}.
\label{eqSets16}
\end{equation}
After a possible rescaling ${\bf A} \mapsto \mu {\bf A}$ $(\mu >0)$, we obtain (\ref{eqSets15}) and (\ref{eqSets15.1}) follows from (\ref{eqSets16}).

Conditions (\ref{eqSets15},\ref{eqSets15.1}) mean that
\begin{equation}
{\bf A} + i (\alpha_0 -1) {\bf \Omega} \ge 0,
\label{eqSets18}
\end{equation}
while
\begin{equation}
{\bf A} + i \alpha_0 {\bf \Omega}  \ngeq 0.
\label{eqSets19}
\end{equation}
On the other hand, we also have
\begin{equation}
{\bf A} + i \gamma_0^{\frac{1}{2n}}  {\bf \Sigma_2} \ge 0.
\label{eqSets21}
\end{equation}
Next, let $h^{\sigma}$ be a non-Gaussian pure state with Planck constant $\alpha_0$, that is (cf. item 6 of Theorem \ref{TheoremPropertiesNWSpectra})
\begin{equation}
 \mathcal{W}^{{\bf J}} (\mathcal{F} h^{\sigma})=  \left\{- \alpha_0,\alpha_0 \right\}  .
\label{eqSets21.1}
\end{equation}
Given ${\bf S} \in \mathcal{D} (n; \omega)$, let $h^{\omega} (z) = h^{\sigma} ({\bf S^{-1}} z)$. From Remark \ref{Remark5.1}:
\begin{equation}
 \left\{- \alpha_0, \alpha_0 \right\} \subset \mathcal{W}^{{\bf \Omega}} (\mathcal{F} h^{\omega}) .
\label{eqSets21.2}
\end{equation}
If $\mathcal{G}_{{\bf A}}$ is a Gaussian with covariance matrix ${\bf A}$, then from (\ref{eqSets18}-\ref{eqSets21}):
\begin{equation}
(\alpha_0 -1, {\bf \Omega})  , (\gamma_0^{\frac{1}{2n}}, {\bf \Sigma_2}) \in \mathcal{W} (\mathcal{F} \mathcal{G}_{{\bf A}}),
\label{eqSets22}
\end{equation}
but
\begin{equation}
(\alpha_0, {\bf \Omega}) \notin \mathcal{W} (\mathcal{F} \mathcal{G}_{{\bf A}}).
\label{eqSets23}
\end{equation}
Given (\ref{eqSets23}), we can always find $h^{\sigma} \notin  \mathcal{L}$ such that
\begin{equation}
g_4 = \mathcal{G}_{{\bf A}} \star h^{\omega} \notin \mathcal{L}.
\label{eqSets24}
\end{equation}
Indeed, we have for ${\bf S} \in \mathcal{D} (n; \omega)$:
\begin{equation}
\begin{array}{c}
g_4 ({\bf S} z) =(\mathcal{G}_{{\bf A}} \star h^{\omega}) ({\bf S} z) = \int_{\mathbb{R}^{2n}} \mathcal{G}_{{\bf A}} ({\bf S} z- z^{\prime}) h^{\omega} (z^{\prime} )d z^{\prime} =\\
\\
=\int_{\mathbb{R}^{2n}} \mathcal{G}_{{\bf A}} ({\bf S}( z- z^{\prime \prime })) h^{\omega} ({\bf S} z^{\prime \prime}) d z^{\prime \prime} =\\
\\
=\int_{\mathbb{R}^{2n}} \mathcal{G}_{{\bf S}^{-1}{\bf A}({\bf S}^{-1})^T} ( z- z^{\prime \prime}) h^{\sigma}(z^{\prime \prime}) d z^{\prime \prime}.
\end{array}
\label{eqSets25}
\end{equation}
Since ${\bf S}^{-1}{\bf A}({\bf S}^{-1})^T + i \alpha_0 {\bf J} = {\bf S}^{-1} \left({\bf A} + i\alpha_0 {\bf \Omega} \right) ({\bf S}^{-1})^T$, from (\ref{eqSets19}) this is not a positive matrix. Consequently, the Gaussian $\mathcal{G}_{{\bf S}^{-1}{\bf A}({\bf S}^{-1})^T} ( z)$ is not a $\sigma$-Wigner function with Planck constant $\alpha_0$. It is a well known fact that, under these circumstances \cite{Janssen1,Janssen2}, the function $h^{\sigma} \notin \mathcal{L}$, which is a Wigner function with Planck constant $\alpha_0$ can always be chosen so that $(\mathcal{G}_{{\bf A}} \star h^{\omega}) ({\bf S} z) = (\mathcal{G}_{{\bf S}^{-1}{\bf A}({\bf S}^{-1})^T} \star h^{\sigma}) (z) \notin \mathcal{L}$.

We thus have for the function
\begin{equation}
g_4 (z) = (\mathcal{G}_{{\bf A}} \star h^{\omega} ) (z)
\label{eqSets26}
\end{equation}
that $g_4 \notin \mathcal{L}$.

Next, notice that ${\bf S}^{-1}{\bf A}({\bf S}^{-1})^T  + i (\alpha_0-1) {\bf J} \ge 0$, and so $(\alpha_0-1, {\bf J}) \in \mathcal{W} (\mathcal{F}\mathcal{G}_{{\bf S}^{-1}{\bf A}({\bf S}^{-1})^T})$. Since $(1-\alpha_0, {\bf J}) \oplus (\alpha_0, {\bf J}) = (1, {\bf J})$, we conclude from (\ref{eqSets25}) and Theorem \ref{TheoremQuantumConditions5} that $g_4 ({\bf S} z)$ is of the $(1, {\bf J})$-positive type. And thus $g_4 (z) \in \mathcal{F}^{\omega}$.

It remains to prove that $g_4 (z) \in \mathcal{F}^{\sigma}$. Since $(\gamma_0^{\frac{1}{2n}}, {\bf \Sigma_2}) \oplus (\alpha_0, {\bf \Omega}) = (1, {\bf J})$, the result follows again from Theorem \ref{TheoremQuantumConditions5}.
\end{proof}

We conclude this subsection by applying the previous techniques to the following representability problem. This completes the discussion initiated in \cite{Dias5} and addressed also in Theorem \ref{TheoremSymplecticCovariance}.

\begin{theorem}\label{TheoremRepresentabilityProblem}
Let $\omega_1 (z,z^{\prime})= z \cdot {\bf \Omega_1^{-1}} z^{\prime}$ be a symplectic form and suppose that ${\bf M} \in Gl(2n; \mathbb{R})$ is not $\omega_1$-symplectic nor $\omega_1$-anti-symplectic. Define $\omega_2 (z,z^{\prime})= z \cdot {\bf \Omega_2^{-1}} z^{\prime}$ where ${\bf \Omega_2}=\frac{1}{\alpha} {\bf M \Omega_1 M^T}$, with $\alpha = | \det {\bf M}|^{\frac{1}{n}}$. Consider again the linear operator $\mathcal{U}_{{\bf M}}$ as defined in (\ref{eqIntroduction3}). Then there exist $W^{\omega_1} \rho \in \mathcal{F}^{\omega_1} $ such that $\mathcal{U}_{{\bf M}}(W^{\omega_1} \rho) \in \mathcal{F}^{\omega_1} \cap \mathcal{F}^{\omega_2}$ and $W^{\omega_1} \rho^{\prime} \in \mathcal{F}^{\omega_1} $ such that $\mathcal{U}_{{\bf M}}(W^{\omega_1} \rho^{\prime}) \in \mathcal{F}^{\omega_2} \backslash \mathcal{F}^{\omega_1}$.
\end{theorem}

\begin{proof}
We follow the strategy used in the proof of the previous theorem for the cases $\mathcal{A}_5$ and $\mathcal{A}_7$. We just have to be cautious because we do not have necessarily $| \det {\bf M}|=1$.

So let $W^{\omega_1} \rho= \mathcal{G}_{{\bf A}}$ be a Gaussian (\ref{eqGaussian1}) with covariance matrix ${\bf A}$ such that
\begin{equation}
{\bf A} + \frac{i}{2}{\bf \Omega_1} \ge 0.
\label{eqTheoremRepresentabilityProblem1}
\end{equation}
Since
\begin{equation}
\mathcal{U}_{{\bf M}}( \mathcal{G}_{{\bf A}})= \mathcal{G}_{{\bf M^{-1}}{\bf A}({\bf M^{-1}})^T},
\label{eqTheoremRepresentabilityProblem2}
\end{equation}
we conclude that $\mathcal{U}_{{\bf M}}(W^{\omega_1} \rho) \in \mathcal{F}^{\omega_1}$ if and only if
\begin{equation}
{\bf M^{-1}}{\bf A}({\bf M^{-1}})^T + \frac{i}{2}{\bf \Omega_1} \ge 0 \Leftrightarrow {\bf A} + \frac{i \alpha}{2}{\bf \Omega_2} \ge 0.
\label{eqTheoremRepresentabilityProblem3}
\end{equation}
Since ${\bf M}$ is not a $\omega_1$-symplectic or a $\Omega_2$-antisymplectic matrix, then $\omega_2 \ne \pm \omega_1$. The matrix ${\bf A}$ satisfies (\ref{eqTheoremRepresentabilityProblem1}) and (\ref{eqTheoremRepresentabilityProblem3}) if and only if $\lambda_{\omega_1,1}({\bf A}) \ge \frac{1}{2}$ and $\lambda_{\omega_2,1}({\bf A}) \ge \frac{\alpha}{2}$. From the case $\mathcal{A}_7$ of the previous theorem we already know how to construct such a matrix after a possible rescaling.

Regarding the function $W^{\omega_1} \rho^{\prime}$ we assume that it is again a Gaussian $\mathcal{G}_{{\bf B}}$ with covariance matrix ${\bf B}$. It is a $\omega_1$-Wigner function if and only if
\begin{equation}
{\bf B} + \frac{i}{2}{\bf \Omega_1} \ge 0,
\label{eqTheoremRepresentabilityProblem4}
\end{equation}
and $\mathcal{U}_{{\bf M}}(W^{\omega_1} \rho^{\prime}) \in \mathcal{F}^{\omega_2} \backslash \mathcal{F}^{\omega_1}$ provided
\begin{equation}
 {\bf B} + \frac{i \alpha}{2}{\bf \Omega_2} \ngeqslant 0.
\label{eqTheoremRepresentabilityProblem5}
\end{equation}
Again, this happens if and only if $\lambda_{\omega_1,1}({\bf B}) \ge \frac{1}{2}$ and $\lambda_{\omega_2,1}({\bf B}) < \frac{\alpha}{2}$. From the case $\mathcal{A}_5$ of the previous theorem, we know that such a matrix exists.
\end{proof}

\section*{Acknowledgements}

We would like to thank Catarina Bastos for providing Figure 1. The
work of N.C. Dias and J.N. Prata is supported by the COST Action
1405 and by the Portuguese Science Foundation (FCT) under the
grant PTDC/MAT-CAL/4334/2014.

\pagebreak

***************************************************************

\textbf{Author's addresses:} \vspace{.5cm}

\textbf{Nuno Costa Dias and Jo\~ao Nuno Prata: }Escola Superior
N\'autica Infante D. Henrique. Av. Eng. Bonneville Franco,
2770-058 Pa\c{c}o d'Arcos, Portugal and Grupo de F\'{\i}sica
Matem\'{a}tica, Universidade de Lisboa, Av. Prof. Gama Pinto 2,
1649-003 Lisboa, Portugal

***************************************************************


\begin{thebibliography}{99}

\bibitem{Bastiaans} M.J. Bastiaans: Wigner distribution function and its
application to first order optics. J. Opt. Soc. Amer. 69 (1979) 1710-1716.

\bibitem{Bastos1} C. Bastos, O. Bertolami, N.C. Dias, and J.N. Prata:
Weyl--Wigner Formulation of Noncommutative Quantum Mechanics, J.
Math. Phys. 49 (2008) 072101.

\bibitem{Bastos2} C. Bastos, N.C. Dias, and J.N. Prata: Wigner measures in
noncommutative quantum mechanics, Comm. Math. Phys. 299 (2010),
no.3, 709--740.

\bibitem{Bastos3} C. Bastos, O. Bertolami, N.C. Dias, and J.N. Prata:
Black holes and phase-space noncommutativity. Phys. Rev. D 80 (2009) 124038.

\bibitem{Bastos4} C. Bastos, O. Bertolami, N.C. Dias, and J.N. Prata:
Phase-space noncommutative quantum cosmology. Phys. Rev. D 78 (2008) 023516.

\bibitem{Bastos5} C. Bastos, O. Bertolami, N.C. Dias, and J.N. Prata:
Singularity problem and phase-space noncanonical noncommutativity. Phys. Rev. D 82 (2010) 041502 (R).

\bibitem{Bastos6} C. Bastos, A. Bernardini, O. Bertolami, N.C. Dias, and J.N. Prata:
Entanglement due to noncommutativity in phase space. Phys. Rev. D 88 (2013) 085013.

\bibitem{Bastos7} C. Bastos, O. Bertolami, N.C. Dias, and J.N. Prata:
Violation of the Robertson-Schr\"odinger uncertainty principle and noncommutative quantum mechanics. Phys. Rev. D 86 (2012) 105030.

\bibitem{Bayen} F. Bayen, M. Flato, C. Fronsdal, A. Lichnerowicz, D. Sternheimer: Deformation theory and quantization I. Deformations of symplectic structures. Ann. Phys. (N.Y.) 111 (1978) 61-110.

\bibitem{Bellissard} J. Bellissard, A. van Elst, H. Schulz-Baldes: The non-commutative geometry of the quantum hall effect. J. Math. Phys. 35 (1994) 5373-5451.

\bibitem{Bolonek} K. Bolonek, P. Kosinski: On uncertainty relations in noncommutative quantum mechanics. Phys. Lett. B 547 (2002) 51-54.

\bibitem{Borowiec} A. Borowiec, A. Pachol: Twisted bialgebroids from a Drinfeld twist. J. Phys. A: Math. Theor. 50 (2017) 055205.

\bibitem{Bracken} A. Bracken, G. Cassinelli, J. Wood: Quantum symmetries and the Weyl-Wigner product of group representations. J. Phys. A: Math. Gen. 36(4) (2003) 1033.

\bibitem{Braunstein} S.L. Braunstein, P. van Loock: Quantum information with continuous variables. Rev. Mod. Phys. 77 (2005) 513-577.

\bibitem{Brocker} T. Br\"ocker, R.F. Werner: Mixed states with positive Wigner functions. J. Math. Phys. 36 (1995) 62.

\bibitem{Cannas} A. Cannas da Silva: Lectures on symplectic geometry. Springer (2001).

\bibitem{Carroll} S.M. Carroll, J.A. Harvey, V.A. Kosteleck\'y, C.D. Lane, T. Okamoto: Noncommutative field theory and Lorentz violation. Phys. Rev. Lett. 87 (2001) 141601.

\bibitem{Chowdhuri1} S.H.H. Chowdhuri, and S.T. Ali: Wigner functions for
noncommutative quantum mechanics: a group representation based construction. J. Math. Phys. 56 (2015) 122102.

\bibitem{Chowdhuri2} S.H.H. Chowdhuri, and S.T. Ali: Triply extended group of translations of $\mathbb{R}^4$ as defining group of NCQM: relations to various gauges. J. Phys. A: Math. Theo. 47 (2014) 085301.

\bibitem{Cohen} L. Cohen, P. Loughlin, G. Okopal: Exact and approximate moments of a propagating pulse. J. Mod. Optics 55 (2008) 3349-3358.

\bibitem{Connes1} A. Connes: Noncommutative Geometry. London-New York: Academic Press, 1994.

\bibitem{Delduc} F. Delduc, Q. Duret, F. Gieres, M. Lefranc{c}ois: Magnetic fields in noncommutative quantum mechanics. J. Phys. Conf. Ser. 103 (2008) 012020.

\bibitem{Demetrian} M. Demetrian, D. Kochan: Quantum mechanics on noncommutative plane. Acta Phys. Slov. 52 (2002) 1.

\bibitem{Dias1} N.C. Dias, M. de Gosson, F. Luef, J.N. Prata: A
Deformation Quantization Theory for Non-Commutative Quantum
Mechanics, J. Math. Phys. 51 (2010) 072101 (12 pages).

\bibitem{Dias2} N.C. Dias, M. de Gosson, F. Luef, J.N.
Prata: A pseudo--differential calculus on non--standard symplectic
space; spectral and regularity results in modulation spaces. J.
Math. Pures Appl. 96 (2011) 423--445

\bibitem{Dias5} N.C. Dias, M. de Gosson, J.N.
Prata: Maximal covariance group of Wigner transforms and pseudo-differential operators. Proc.
Amer. Math. Soc. 142 (2014) 3183--3192

\bibitem{Dias3} N.C. Dias, J.N. Prata: Admissible states in quantum phase space.
Ann. Phys.(N. Y.), 313 (2004) 110--146.

\bibitem{Dias4} N.C. Dias, J.N. Prata: The Narcowich-Wigner spectrum of a pure state.
Rep. Math. Phys. 63 (2009) 43--54.

\bibitem{Dias6} N.C. Dias, J.N. Prata: Exact master equation for a noncommutative Brownian particle. Ann. Phys. (N.Y.) 324 (2009) 73-96.

\bibitem{Douglas} M.R. Douglas, N.A. Nekrasov: Noncommutative field theory. Rev. Mod. Phys. 73 (2001) 977.

\bibitem{Dubin} D. Dubin, M. Hennings, T. Smith: Mathematical aspects of Weyl quantization. Singapore: World Scientific, 2000.

\bibitem{Duval} C. Duval, P.A. Horvathy: Exotic galilean symmetry in the noncommutative plane and the Landau effect. J. Phys. A: Math. Gen. 34 (2001) 10097.

\bibitem{Ellinas} D. Ellinas, A.J. Bracken: Phase-space-region operators and the Wigner function: geometric constructions and tomography. Phys. Rev. A 78 (2008) 052106.

\bibitem{Fedosov1} B. Fedosov: A simple geometric construction of deformation quantization. J. Diff. Geom. 40 (1994) 213.

\bibitem{Fedosov2} B. Fedosov: Deformation quantization and index theory. Berlin: Akademie Verlag, 1996.

\bibitem {Folland}G. B. Folland: Harmonic Analysis in Phase space. Annals of
Mathematics studies, Princeton University Press, Princeton, N.J. (1989).

\bibitem{Gamboa} J. Gamboa, M. Loewe, J.C. Rojas: Noncommutative quantum mechanics. Phys. Rev. D 64 (2001) 067901.

\bibitem{Obregon} H. Garc\'{\i}a-Compe\'an, O. Obreg\'on, C. Ram\'{\i}rez: Noncommutative quantum cosmology. Phys. Rev. Lett. 88 (2002) 161301.

\bibitem{Giulini} D.Giulini, E. Joos, C. Kiefer, J. Kupsch, I.-O. Stamatescu, H.D. Zeh: Decoherence and the appearence of a classical world in quantum theory. Springer (1996).

\bibitem {Gosson1}M. de Gosson. Symplectic Geometry and Quantum Mechanics, Birkh\"{a}user, Basel, series \textquotedblleft Operator Theory:
Advances and Applications\textquotedblright\ (subseries:
\textquotedblleft Advances in Partial Differential
Equations\textquotedblright), Vol. 166 (2006)

\bibitem {Gosson2}M. de Gosson. Symplectic Methods in Harmonic Analysis ans in Mathematical Physics, Birkh\"{a}user, Basel, series \textquotedblleft Pseudo-Differential Operators:
Theory and Applications\textquotedblright, Vol. 7 (2010).

\bibitem{Gosson3} M. de Gosson: Maslov indices on the metaplectic group Mp(n). Ann. Inst. Four. 40 (1990) 537-555.

\bibitem{Gosson4} M. de Gosson: Symplectic covariance properties for Shubin and Born-Jordan pseudo-differential operators. Trans. Amer. Math. Soc. 365 (2013) 3287-3307.

\bibitem{Grubb} G. Grubb: Distributions and operators. Berlin-Heidelberg-New York: Springer, 2009.

\bibitem{Henneaux} M. Henneaux, C. Teitelbaum: Quantization of gauge systems. Princeton University Press (1992).

\bibitem{Hormander} L. H\"ormander: The analysis of linear partial differential operators I. Berlin-Heidelberg-New York: Springer-Verlag, 1983.

\bibitem{Horvathy} P.A. Horvathy: The noncommutative Landau problem. Ann. Phys. (N.Y.) 299 (2002) 128.

\bibitem{Hudson} R.L. Hudson: When is the Wigner quasi-probability density non-negative? Rep. Math. Phys. 6 (1974) 249-252.

\bibitem{Janssen1} A.J.E.M. Janssen: A note on Hudson's Theorem about functions with nonnegative Wigner distributions. SIAM J. Math. Anal. 15 (1984) 170.

\bibitem{Janssen2} A.J.E.M. Janssen: Positivity of weighted Wigner distributions. SIAM J. Math. Anal. 12 (1981) 752.

\bibitem{Jiang} J.-J. Jiang, S.H.H. Chowdhuri: Deformation of noncommutative quantum mechanics. Arxiv: math-ph: 1603.05805.

\bibitem{Kastler} D. Kastler, The $C^{\ast} $-algebras of a free boson field. Commun. Math. Phys. 1 (1965) 14.

\bibitem{Kontsevich} M. Kontsevich: Deformation quantization of Poisson manifolds. Lett. Math. Phys. 66 (2003) 157.

\bibitem{Leray} J. Leray: Lagrangian analysis and quantum mechanics. A mathematical structure related to asymptotic expansions and the Maslov index. MIT Press, Cambridge, Mass. (1981).

\bibitem{Lions} P.L. Lions, T. Paul: Sur les mesures de Wigner. Rev. Mat. Iberoamer. 9 (1993) 553-618.

\bibitem{Littlejohn} R.G. Littlejohn: The semiclassical evolution of wave packets. Phys. Rep. 138 (1986) 193.

\bibitem{Loughlin} P. Loughlin, L. Cohen: Approximate wavefunction from approximate non-representable Wigner distributions. J. Mod. Optics 55 (2008) 3379-3387.

\bibitem{Loupias} G. Loupias, S. Miracle-Sole: $C^{\ast}$-alg\`ebres des syst\`emes canoniques. Ann. Inst. H. Poincar\'e 6 (1967) 39.

\bibitem{Madore} J. Madore: An introduction to noncommutative differential geometry and its physical applications. 2nd Edition. Cambridge: Cambridge University Press (2000).

\bibitem{Malekolkalami} B. Malekolkalami, M. Farhoudi: Noncommutative double scalar fields in FRW cosmology as cosmical oscillators. Class. Quant. Grav. 27 (2010) 245009.

\bibitem{Martinez} A. Martinez: An introduction to semiclassical and microlocal analysis. Springer (2002).

\bibitem{Monreal} L. Monreal, P. Fern\'andez de C\'ordoba, A. Ferrando, J.M. Isidro: Noncommutative space and the low-energy physics of quasicrystrals. Int. J. Mod. Phys. A 23 (2008) 2037-2045.

\bibitem{Moyal} J.E. Moyal: Quantum mechanics as a statistical theory. Proc. Cambr. Phil. Soc. 45 (1949) 99-124.

\bibitem{Nair} V.P. Nair, A.P. Polychronakos: Quantum mechanics on the noncommutative plane and sphere. Phys. Lett. B 505 (2001) 267.

\bibitem{Nakamura1} M. Nakamura: Star-product description of qunatization in second-class constraint systems. Arxiv: math-ph/1108.4108 (2011).

\bibitem{Nakamura2} M. Nakamura: Canonical structure of noncommutative quantum mechanics as constraint system. Arxiv: hep-th/1402.2132 (2014).

\bibitem{Narcowich1} F.J. Narcowich: Conditions for the convolution of two Wigner distributions to be itself a Wigner distribution. J. Math. Phys. 29 (1988) 2036.

\bibitem{Nicolini} P. Nicolini, A. Smailagic, E. Spallucci: Noncommutative geometry inspired Schwarzschild black hole. Phys. Lett. B 632 (2006) 547-551.

\bibitem{Pool} J.C. Pool: Mathematical aspects of the Weyl correspondence. J. Math. Phys. 7 (1966) 66.

\bibitem{Reed} M. Reed, B. Simon: Methods in Modern Mathematical Physics. Vol. 1: Functional Analysis. Academic Press. Elsevier (1980).

\bibitem{Rovelli} C. Rovelli: Quantum gravity. Cambridge University Press (2004).

\bibitem{Seiberg} N. Seiberg, E. Witten: String theory and noncommutative geometry. JHEP 9909 (1999) 032.

\bibitem{Shale} D. Shale: Linear symmetries of free boson fields. Trans. Amer. Math. Soc. 103 (1062) 149-167.

\bibitem{Simon} R. Simon: Peres-Horodecki separability criterion for continuous variable systems. Phys. Rev. Lett. 84 (2000) 2726.

\bibitem{Soto} F. Soto, P. Claverie: When is the Wigner function of multidimensional systems nonnegative? J. Math. Phys. 24 (1983) 97.

\bibitem{Szabo} R.J. Szabo: Quantum field theory on noncommutative spaces. Phys. Rept. 378 (2003) 207-299.

\bibitem{Tosiek} J. Toseik: The eigenvalue equation for a 1-D Hamilton function in deformation quantization. Phys. Lett. A 376 (2012) 2023-2031.

\bibitem{Vilela} R. Vilela Mendes: Deformations, stable theories and fundamental constants. J. Phys. A: Math. Gen. 27 (1994) 8091-8104.

\bibitem{Weil} A. Weil: Sur certains groupes d'op\'erateurs unitaires. Acta. Math. 111 (1964) 143-211.

\bibitem{Werner} R.F. Werner, M.M. Wolf: Bound entangled Gaussian states. Phys. Rev. Lett. 86 (2001) 3658.

\bibitem{Wigner} E. Wigner: On the quantum correction for for thermodynamics equilibrium. Phys. Rev. 40 (1932) 749.

\bibitem{Wilde} M. Wilde, P. Lecomte: Existence of star-products and of formal deformations of the Poisson Lie algebra of arbitrary symplectic manifolds. Lett. Math. Phys. 7 (1983) 487.

\bibitem{Williamson} J. Williamson: On the algebraic problem concerning the normal forms of linear dynamical systems. Amer. J. Math. 58 (1936) 141-163.

\bibitem {Wong}M. W. Wong. Weyl Transforms\textit{.} Springer (1998).

\bibitem{Zworski} M. Zworski: Semi-classical analysis. Graduate studies in mathematics 138. American Mathematical Society, Providence (2012).

\end{thebibliography}
\end{document}